\documentclass[12pt,english]{article}
\usepackage[top=1in, bottom=1in, left=1in, right=1in]{geometry}
\usepackage{epsfig,fullpage,wrapfig,amsmath,amssymb,float,titlesec,multirow,enumitem,hhline,makecell}
\usepackage[table]{xcolor}
\usepackage[T1]{fontenc}
\usepackage[colorlinks=true, linkcolor=black, citecolor=black, urlcolor=black, filecolor=black]{hyperref}
\usepackage{bookmark}
\usepackage{xurl}
\usepackage{titling,color,booktabs,caption,subcaption,tcolorbox, amsthm}
\usepackage{setspace}
\usepackage{rotating}

\usepackage{titlesec}
\titlespacing*{\section}{0pt}{10pt plus 2pt minus 2pt}{6pt plus 1pt minus 1pt}
\titlespacing*{\subsection}{0pt}{8pt plus 2pt minus 2pt}{4pt plus 1pt minus 1pt}
\titlespacing*{\subsubsection}{0pt}{6pt plus 2pt minus 2pt}{3pt plus 1pt minus 1pt}
\makeatletter
\def\@listI{\leftmargin\leftmargini \parsep 2pt \topsep 4pt \itemsep 2pt}
\makeatother

\newcommand{\E}{\mathbb{E}}
\newcommand{\Prob}{\mathbb{P}}
\newcommand{\R}{\mathbb{R}}
\newcommand{\argmax}{\operatorname{argmax}}

\usepackage[
  style=authoryear,
  maxcitenames=2,
  mincitenames=1,
  maxbibnames=99,
  minbibnames=99,
  uniquelist=false,
  uniquename=false
]{biblatex}
\newtheorem{theorem}{Theorem}

\newtheorem{proposition}{Proposition}

\newtheorem{definition}{Definition}
\newtheorem{assumption}{Assumption}
\newtheorem{example}{Example}

\usepackage{algorithm}
\usepackage{algorithmic}

\renewcommand{\hat}{\widehat}
\renewcommand{\tilde}{\widetilde}

\begin{document}

\title{\textbf{\large Compliance Moral Hazard and the Backfiring Mandate}}
\author{Jian Ni, Lecheng Zheng, John R Birge\thanks{Jian Ni: Pamplin College of Business, Virginia Tech, Email: jiann@vt.edu; Lecheng Zheng: Pamplin College of Business, Virginia Tech, Email: lecheng@vt.edu; John R Birge: Booth School of Business, University of Chicago, Email: John.Birge@chicagobooth.edu}}
\date{}

\maketitle
\thispagestyle{empty}

\begin{singlespace}
\begin{abstract}
\sloppy
Competing firms that share a population of risky customers face a decentralized risk detection problem in which each firm holds fragmentary information whose aggregation would generate social value, but private incentives impede truthful sharing. We develop a dynamic mechanism design framework for this setting and identify three strategic frictions that distinguish it from classical mechanism design with decentralized information: compliance moral hazard, adversarial adaptation, and information destruction through intervention. A temporal value assignment (TVA) mechanism credits firms using a strictly proper scoring rule applied to discounted verified outcomes; under stated assumptions, TVA implements truthful posterior reporting as a Bayes--Nash equilibrium (uniquely optimal at each edge in large federations, with $O(1/m)$ shading in finite systems). A network Shapley characterization shows that under edge-additive coalition value, each firm's marginal contribution is proportional to its weighted cross-firm interaction degree, yielding a sharp prescription for coalition design that prioritizes inter-firm volume over firm size. Embedding TVA in a model of competition among firms, we establish a welfare ordering across four regulatory regimes (autarky, voluntary federation, mandated full sharing, TVA) and identify conditions under which information-sharing mandates without compatible incentive design reduce welfare below autarky: a ``backfiring mandate.'' We illustrate the framework on a 1.4M-transaction synthetic anti-money-laundering benchmark; the same machinery extends to platform fraud, cybersecurity threat intelligence, and supply chain risk detection.
\end{abstract}

\end{singlespace}

\newpage


\section{Introduction}

Decentralized risk detection across competing firms is a market design problem of growing practical importance. Payment platforms detecting cross-merchant fraud, fintech lenders assessing credit risk across origination channels, banks coordinating on illicit financial flows, supply chain participants exchanging counterparty risk signals, and cybersecurity teams sharing threat intelligence all face a common structure: each firm holds fragmentary information whose aggregation would generate substantial social value, but private incentives impede truthful sharing. Information aggregation by a benevolent social planner is unavailable; the market designer must instead operate through a decentralized protocol in which firms train local models on private data, submit reports about their detected risk, and intervene on flagged accounts. The resulting system combines features of multi-sender information aggregation \parencite{myerson1983efficient}, dynamic mechanism design \parencite{bergemann2010dynamic, pavan2014dynamic}, and platform governance, with welfare consequences for all participating firms and their downstream customers.

This paper develops a dynamic mechanism design framework for this setting. Three strategic frictions distinguish the problem from classical mechanism design with decentralized information and require simultaneous treatment. The first is \emph{compliance moral hazard}: firms that share risk signals may trigger costly investigations or reveal detection capabilities to competitors, creating free-riding incentives \parencite{begley2017strategic}. When detection benefits are shared through federation but compliance costs are borne locally, firms face incentives to underreport. Competitive pressure intensifies this moral hazard by raising the opportunity cost of flagging risky customers; a related channel is documented for deposit insurance in \textcite{bao2017could}. The second is \emph{adversarial adaptation}: sophisticated bad actors observe intervention patterns and restructure their behavior in response, requiring detection policies robust to strategic manipulation. This connects to the literature on strategic classification \parencite{hardt2016strategic, dong2018strategic}, with the additional complication that the adversary's adaptation is endogenous to the entire system of firms' reporting and the regulator's intervention. The third is \emph{information destruction through intervention}: acting on detected risk permanently removes nodes and edges from the observation network, creating an exploration--exploitation tradeoff where aggressive intervention improves immediate security but degrades future learning. This echoes optimal stopping problems in the bandit literature \parencite{gittins1979bandit, whittle1988restless}.

We propose a constructive mechanism, \emph{Temporal Value Assignment} (TVA), and develop four results that together characterize when and how decentralized risk detection can be welfare-improving. We illustrate the framework throughout using anti-money laundering (AML) in cross-border banking networks, where the information fragmentation problem is acute and well-documented: \$800 billion to \$2 trillion is laundered annually through legitimate financial channels, less than 1\% is ever recovered \parencite{unodc2025}, and global compliance spending has risen to roughly \$200 billion annually \parencite{amoako2025exploring}. AML offers a setting in which the strategic frictions are simultaneously present, regulatory mandates are explicit (e.g., the EU 6th AML Directive 2018/1673 and FinCEN \S314(b) information-sharing programs), and a public benchmark dataset is available. The same machinery applies to platform fraud detection, cybersecurity threat intelligence, supply chain risk, and other settings where competing firms must collectively detect risky agents without giving up proprietary data.\footnote{The empirical implementation of the detection framework, including the federated graph neural network architecture, hierarchical reinforcement learning intervention policy, and large-scale validation, is developed in a separate companion paper, \textcite{zheng2025networked}. The two papers have non-overlapping primary contributions: the companion develops the empirical machinery, and \emph{this paper} develops the mechanism design framework, the welfare analysis, and the network coalition theory. We use the same benchmark dataset for illustration of theoretical predictions, with quantitative results reported here serving only to illustrate (not to establish) the theoretical claims.}

We develop a formal framework in which heterogeneous firms hold fragmented signals and participate in a federated learning protocol that aggregates local models without sharing raw transaction records. The strategic question is whether parameter reports faithfully reflect local risk assessments or are distorted to minimize compliance costs. We use graph neural networks to represent each firm's local transaction network (architecture details in the Online Appendix), and evaluate the framework in simulation using the IBM AML synthetic benchmark \parencite{altman2023realistic}, comprising 1.4 million transactions across seven markets. Four main results organize the analysis.

The first result is the \emph{incentive design contribution}. We introduce TVA, which credits firms using a strictly proper scoring rule \parencite{gneiting2007strictly} applied to discounted verified outcomes. Strict propriety guarantees that, edge by edge, the truthful posterior is the unique maximizer of expected score; combining this pointwise property with a delay structure that places compliance costs on aggregate signals (rather than individual reports) yields the central incentive-compatibility theorem: under stated assumptions, truthful reporting is a Bayes--Nash equilibrium that is uniquely optimal at each edge in large federations, and an $\epsilon$-BNE with $\epsilon = O(1/m)$ in finite federations of size $m$ (Theorem~\ref{thm:truthful}). The mechanism is self-financing from the rent on confirmed illicit cases, requiring no upfront budget.

The second result is the \emph{welfare analysis}. Embedding the detection game in a logit competition model among firms in the tradition of \textcite{keeley1990monopoly} and the structural industrial organization of financial markets \parencite{bao2017could, an2023index}, we establish a welfare ordering across four regulatory regimes (autarky, voluntary federation, mandated full sharing, TVA). The Backfiring Mandate Proposition (Proposition~\ref{thm:welfare_ordering}) identifies conditions under which voluntary federation without incentive design reduces welfare \emph{below} autarky: when competition is sufficiently intense, strategic underreporting produces biased global models worse than honest local models while still imposing compliance costs. Above a threshold of competitive intensity, even mandated full sharing falls below the TVA equilibrium. In calibrated simulations, mandatory sharing without TVA reaches only 56\% of first-best welfare, barely above autarky (54\%), while TVA achieves 87\%. This ``backfiring'' result is the paper's main welfare contribution and has direct implications for current AML regulation.

The third result is the \emph{network theory of information value}. Under an edge-additive coalition value, each firm's Shapley contribution to collective detection is proportional to its weighted degree in the inter-firm interaction graph (Proposition~\ref{thm:shapley}). The implication is sharp: coalition design should prioritize firms with high inter-firm interaction volume rather than firms that are large in absolute terms. A path-based extension connecting to Bonacich centrality \parencite{ballester2006s} is outlined in the Online Appendix. The result builds on \textcite{myerson1977graphs} and the network economics literature \parencite{jackson1996strategic, jackson2008social, galeotti2010network, bramoulle2014strategic, jackson2005allocation}.

The fourth result is an \emph{operational intervention extension}. Freezing accounts creates a restless bandit problem \parencite{whittle1988restless} on a network: aggressive intervention reduces immediate losses but degrades future detection capability through node removal. We propose a network-adjusted index-based heuristic (Proposition~\ref{thm:gittins}) and a tractable risk memory approximation that performs close to the oracle benchmark in simulation. The construction integrates online learning under adversarial adaptation \parencite{joulani2013online, arora2012online} with submodular optimization bounds \parencite{nemhauser1978analysis}.

\subsection{Related Literature}
\label{sec:literature}

Our mechanism connects to dynamic mechanism design with marginal-contribution transfers \parencite{bergemann2010dynamic, pavan2014dynamic}: TVA~\eqref{eq:credit} is a delayed contribution-based transfer, structurally related to dynamic pivot and VCG-type mechanisms \parencite{groves1973incentives, clarke1971multipart, vickrey1961counterspeculation} but adapted to a non-stationary environment where intervention alters the state space. The truthful elicitation of probabilistic beliefs draws on strictly proper scoring rules \parencite{gneiting2007strictly}; embedding such a rule in the dynamic transfer yields truthful reporting as a BNE uniquely optimal at each edge in large federations (Theorem~\ref{thm:truthful}). On strategic disclosure, our setting parallels the finance literature on underreporting of regulatory exposure \parencite{begley2017strategic, admati2000forcing} and information design in competitive markets \parencite{kamenica2011bayesian, gentzkow2016competition}, though institutions here jointly construct the information structure rather than a single designer choosing it. The competition channel builds on banking theory showing competitive pressure erodes prudential incentives \parencite{keeley1990monopoly, hellmann2000liberalization} and connects to the broader literature on the industrial organization of financial markets that uses structural demand models to assess welfare under competition \parencite{bao2017could, an2023index}. Empirical work on privacy regulation in financial intermediation provides direct motivation for the leakage-cost primitive in our model: \textcite{doerr2026privacy} document that the California Consumer Privacy Act (CCPA) materially altered the bank--fintech competitive landscape in mortgage lending, consistent with our framework's prediction that the design of disclosure regulation interacts with competition in determining welfare.

The welfare analysis relates to the literature on information sharing among competitors \parencite{vives1990trade, gal1985information} and on voluntary disclosure \parencite{admati2000forcing}. The competition--moral hazard interaction connects to the theoretical banking literature on how competitive pressure erodes prudential incentives \parencite{keeley1990monopoly, hellmann2000liberalization, morrison2005crises}. Our Backfiring Mandate Proposition echoes the finding of \textcite{bao2017could} that well-intentioned regulatory interventions can intensify rather than mitigate moral hazard when competitive pressure is strong: in our setting, mandating federation without incentive design encourages free-riding on others' detection effort in the same way that deposit insurance can reduce banks' incentive for prudent risk management \parencite{keeley1990monopoly, hellmann2000liberalization}. What is novel is the explicit welfare ordering across four regulatory regimes and the identification of network centrality, rather than institution size, as the key determinant of coalition efficiency.

The network Shapley value builds on the foundational work of \textcite{myerson1977graphs} on cooperative games in graph structures and the network economics literature \parencite{jackson1996strategic, jackson2008social}. A path-based extension of our Shapley characterization (Online Appendix) connects to the key-player result of \textcite{ballester2006s}, who show that the agent whose removal most reduces aggregate Nash equilibrium activity is the one with the highest intercentrality. Our core Proposition~\ref{thm:shapley} uses the weighted-degree structure that follows from the edge-additive coalition value, while the path-based extension relates to Bonacich centrality. Further foundations are in \textcite{galeotti2010network}, \textcite{bramoulle2014strategic}, and \textcite{jackson2005allocation}. On the machine learning side, the strategic reporting frictions we study connect to the strategic classification literature \parencite{hardt2016strategic, dong2018strategic} and to federated learning \parencite{mcmahan2017communication, li2020federated}. The graph neural network architecture underlying the detection model builds on foundational GNN work \parencite{kipf2016semi, velickovic2018graph}. The applied graph-based AML detection literature \parencite{weber2019anti, weber2018scalable, zheng2025cluster, zheng2024multiview} treats detection as a pure prediction problem; our contribution is providing game-theoretic foundations in which network position is a first-class strategic variable: institutions choose how much to contribute to the graph-based detection system, and their network centrality determines both their incentive constraints and their coalition value.

The remainder of this paper is organized as follows. Section~\ref{sec:federated} introduces the TVA mechanism, the incentive analysis with a quantitative illustration on the IBM AML synthetic benchmark, and adversarial robustness. Section~\ref{sec:heterogeneous} characterizes coalition formation and the network theory of information value. Section~\ref{sec:competition} embeds the detection game in banking competition and establishes the welfare ordering across regulatory regimes. Section~\ref{sec:intervention} develops the operational intervention extension under additional dynamic assumptions, including simulation evidence. Section~\ref{sec:conclusion} concludes. All proofs are in the Appendix; additional empirical and technical material is in the Online Appendix.


\section{Model and Mechanism Design}

\subsection{Federated Detection Framework}
\label{sec:federated}

The financial system consists of $m$ institutions indexed by $i \in \{1, \ldots, m\}$ operating over discrete time periods $t \in \{1, \ldots, T\}$. Each institution $i$ observes a local transaction network $G^{i,t} = (V^{i,t}, E^{i,t}, X^{i,t}_V, X^{i,t}_E)$, where $V^{i,t}$ is the set of accounts active at institution $i$ in period $t$, $E^{i,t} \subseteq V^{i,t} \times V^{i,t}$ is the set of fund transfers (edges), $X^{i,t}_V \in \R^{|V^{i,t}| \times d_V}$ is the matrix of node (account) features, and $X^{i,t}_E \in \R^{|E^{i,t}| \times d_E}$ is the matrix of edge (transaction) features. Cross-institution transactions create edges between institutions: if account $u \in V^{i,t}$ sends funds to $v \in V^{j,t}$ with $i \neq j$, the edge $(u,v)$ appears in both $E^{i,t}$ and $E^{j,t}$. Each edge $e \in E^{i,t}$ carries a binary label $y_e \in \{0,1\}$ indicating whether it is illicit. Labels are confirmed with delay: the true label for an edge observed at time $t$ may not be determined until $t + \tau$ for some $\tau \geq 0$, reflecting the time required for regulatory investigation, court proceedings, or account seizure to establish that a transaction was illicit.\footnote{In AML practice, label confirmation follows the Suspicious Activity Report (SAR) investigation process in which the institution flags the activity, FinCEN or a national financial intelligence unit reviews it, and a determination is made, typically over weeks to months. The delay $\tau$ captures this institutional timeline.} The class distribution is severely imbalanced: $\Prob(y_e = 1) < 0.01$ in typical financial data.

Each institution $i$ maintains a graph-based detection model $f^i_\theta: G^{i,t} \to [0,1]^{|E^{i,t}|}$ that maps the transaction network to risk scores $\hat{y}^{i,t}_e$ for each edge. We implement $f^i_\theta$ using graph neural networks (GNNs) \parencite{kipf2016semi, DBLP:conf/iclr/XuHLJ19}, which aggregate information from local transaction neighborhoods through learned message-passing, followed by temporal attention across time windows to capture evolving laundering patterns, and an edge-level multi-layer perceptron for illicit transaction classification. To address class imbalance, we adopt focal loss \parencite{lin2017focal}, which dynamically reweights training samples toward hard-to-classify minority cases. Full architectural details are provided in the Online Appendix.

The key feature of the protocol is privacy preservation: each institution $i$ trains its local model $f^i_\theta$ on $G^{i,t}$ for a fixed number of local epochs, then transmits model parameters $\theta^i$ (not raw data) to a central coordinator that aggregates via weighted averaging, $\theta^{\text{global}} = \sum_{i=1}^m w_i \theta^i$, and distributes the updated global model back to all institutions \parencite{mcmahan2017communication}. This process repeats for $R$ communication rounds. Training is supervised: institution $i$ uses the local graph $G^{i,t}$ together with all confirmed labels $\{y_e : e \in E^{i,t'}, t' \leq t\}$ available up to period $t$. Importantly, institutions never share raw transaction data $(G^{i,t}, y^{i,t})$: only model parameters. Yet the parameters institutions submit are themselves chosen strategically, and the strategic incentives behind those choices are the focus of the next subsection.

\subsection{Strategic Environment and Incentive Alignment}

The federated detection protocol described above defines what each institution can submit, but not what each institution will submit. Reporting is a strategic choice: institutions face private benefits from accurate detection (rents from confirmed cases) and private costs from heightened compliance and information leakage. To analyze how these forces interact, we model the interaction as a repeated game among three types of players: \textit{financial institutions} (information senders), \textit{adaptive adversaries} (money launderers), and a \textit{regulator} (decision maker). The game unfolds over discrete time periods $t = 1, \dots, T$, with decentralized information, delayed feedback, and endogenous adaptation by adversaries to past interventions.

\noindent\textbf{Institutions (Information Senders).}
Each institution $i$ observes its local transaction network $G^{i,t}$ and submits a report consisting of model parameters $\vartheta_i \in \Theta$ to a central coordinator. A reporting strategy $\rho_i : G^{i,t} \rightarrow \Theta$ maps local observations to reported parameters. The submitted parameters induce \emph{edge-level posterior probabilities} $\hat{y}^{i,t}_e \equiv f_{\vartheta_i}(e; G^{i,t}) \in [0,1]$ on each edge $e \in E^{i,t}$, which are the mechanism-relevant quantities entering the TVA credit rule~\eqref{eq:credit}. The truthful benchmark strategy is
\[
    \rho_i^*(G^{i,t}) = \theta^i,
    \quad \text{where }
    \theta^i = \arg\min_{\theta} \mathcal{L}(f_\theta, G^{i,t}),
\]
for a suitable loss function $\mathcal{L}$ (e.g., focal loss). Under truthful reporting, the induced posteriors $\hat{y}^{i,t}_e = f_{\theta^i}(e; G^{i,t})$ are the institution's best Bayesian estimates of the illicit probability $\Pr(y_e = 1 \mid G^{i,t}, \theta^i)$: that is, institutions report calibrated beliefs. Strategic deviations, such as reporting $\vartheta_i \neq \theta^i$, manifest as distorted posteriors at the edge level, and the mechanism analysis therefore operates equivalently on parameters $\vartheta_i$ or on the induced posteriors $\hat{y}^{i,t}$. Institutions may deviate by underreporting risk (scaling posteriors below their calibrated values) to free-ride on others' detection efforts, injecting parameter noise (which translates to noise in posteriors), or withholding participation.\footnote{That model parameters can expose local training data is well-established in the federated learning literature; see \textcite{zhu2019deep} and \textcite{melis2019exploiting} on gradient inversion and membership inference attacks. The mutual information term $\kappa_i \cdot I(\theta^i; G^{i,t})$ in equation~\eqref{eq:bank_utility} provides a formal measure of this leakage.}

\noindent\textbf{Adversaries (Money Launderers).}
Adversaries observe the history of regulatory interventions $\{a^1, \ldots, a^t\}$, where $a^s$ indicates which transaction edges were frozen at time $s$. Using this feedback, sophisticated launderers infer the detection policy and adapt through transaction splitting, temporal manipulation, route switching, and account replacement. Formally, the adversary chooses the next-period transaction graph $G^{t+1}$ to solve
\begin{equation}
\max_{G^{t+1}} \;
\mathbb{E}[\text{laundered value}]
    - \Prob(\text{detected} \mid \text{policy history}) \cdot C,
    \label{eq:adversary}
\end{equation}
where $C$ denotes the cost of detection, including account freezing and asset seizure.

\noindent\textbf{Regulator (Decision Maker).}
The regulator observes aggregated risk scores $\{\hat{y}_e^t\}$ produced by the global model and selects intervention actions $a_e \in \{\text{monitor}, \text{freeze}\}$ for suspicious transaction edges. Freezing prevents immediate illicit losses but removes edges from future transaction graphs, degrading subsequent learning.

Within each period $t$, the interaction proceeds as follows: (1)~institutions observe local transaction graphs and submit (possibly distorted) reports; (2)~the coordinator aggregates reports and updates the global model, and the regulator computes risk scores and selects interventions; (3)~adversaries observe interventions and adapt for period $t+1$; (4)~true labels are revealed with delay.

\noindent\textbf{Social Planner Benchmark.}
Consider a social planner with access to all local transaction graphs $\{G^{i,t}\}_{i=1}^m$. The planner jointly optimizes detection and intervention to maximize discounted social welfare:
\begin{equation}
\max_{\theta, \{a_e^t\}}
\mathbb{E}\left[
\sum_{t=1}^T \gamma^t
\left(
\sum_{e: a_e^t = \text{freeze}} C_e y_e
- \alpha_1 \sum_{e: a_e^t = \text{freeze}} (1 - y_e)
- \alpha_2 \sum_{e: a_e^t = \text{monitor}} y_e
\right)
\Bigg| \mathcal{F}_t \right],
\label{eq:social}
\end{equation}
where $\theta$ denotes the global detection model parameters that determine the risk scores $\hat{y}_e^t$ used in intervention decisions, $a_e^t \in \{\text{monitor}, \text{freeze}\}$ is the intervention action for edge $e$ at time $t$, $\mathcal{F}_t$ is the information available at time $t$ comprising all observed graphs and confirmed labels up to $t$, $C_e$ is the prevented illicit value from freezing edge $e$, and $\alpha_1$, $\alpha_2$ are Type I and Type II error costs.

In practice, however, all local transaction graphs are not shared across institutions due to privacy and competitive concerns. Each institution $i$ instead maximizes a private utility function:
\begin{equation}
u_i(\rho_i, \rho_{-i})
=
\mathbb{E}[\text{detection quality}]
- \beta_i \cdot \mathbb{E}[\text{local compliance costs}]
- \kappa_i \cdot I(\theta^i; G^{i,t}),
\label{eq:bank_utility}
\end{equation}
where $I(\theta^i; G^{i,t})$ denotes the mutual information between reported parameters and local data, capturing information leakage costs. This divergence between institutions and the social planner creates incentives for underreporting, delayed signaling, and strategic noise injection, leading to systematically suboptimal collective detection.

Truthful reporting imposes two costs on institution $i$: (i)~a \textit{direct compliance cost}, since strong suspicion signals trigger more local freezing decisions and the associated investigations; and (ii)~an \textit{information leakage cost}, since reporting $\theta^i$ reveals local transaction patterns to competitors and potentially to adversaries. Meanwhile, detection benefits are dispersed across all institutions through reduced systemic risk and regulatory penalties. This creates incentives for strategic underreporting.

\begin{example}[Strategic Underreporting]
Consider institution $A$ in a federation of $m$ banks, with three edges having true posteriors $q_e \in \{0.6, 0.7, 0.8\}$. Under the Brier score, truthful reporting yields expected credit $2.39 C$; unilateral underreporting to $[0.48, 0.56, 0.64]$ reduces this to $2.357 C$ by strict propriety. Under Assumptions~\ref{ass:compliance_cost} and~\ref{ass:no_dominant}, underreporting reduces global detection sensitivity on $A$'s edges, lowering expected investigations and hence compliance cost; however, because $A$'s unilateral deviation shifts the aggregated signal by only $O(1/m)$, this induced compliance saving is itself $O(c_A/m)$ to first order and vanishes in large federations. For moderately large $m$, the first-order credit loss $0.033 C$ dominates the $O(1/m)$ compliance saving. Coordinated underreporting by all institutions would collapse aggregate detection quality, a collective-action failure that TVA prevents by making individual truthfulness a strict best response at each edge.
\end{example}

This example shows that decentralized institutions face a misalignment between marginal social benefits and marginal private costs, yielding strategic underreporting as a rational equilibrium. The key question is: \textit{can we design a mechanism that aligns institutional incentives with the social objective?}

\noindent\textbf{Temporal Value Assignment.}
We introduce a \textit{Temporal Value Assignment} (TVA) mechanism \parencite{sutton1984temporal} that explicitly rewards institutions for early and accurate warnings. The coordinator maintains a dynamic credit account $\pi_t^i$ for each institution $i$. When a transaction edge $e$ is confirmed illicit at time $t_{\text{confirm}}$, institution $i$ receives credit based on its historical predictions:
\begin{equation}
\pi^i_{t_{\text{confirm}}}
=
\pi^i_{t_{\text{confirm}}-1}
+
\sum_{t < t_{\text{confirm}}}
\gamma^{t_{\text{confirm}} - t}
\cdot
C_e
\cdot
S\!\left(\hat{y}^{i,t}_e, y_e\right),
\label{eq:credit}
\end{equation}
where $\gamma \in (0,1)$ discounts delayed signals, $y_e \in \{0,1\}$ is the confirmed label, and $S: [0,1] \times \{0,1\} \to \R$ is a \emph{strictly proper scoring rule} \parencite{gneiting2007strictly}. Our baseline specification is the (negative) Brier score $S(\hat{y}, y) = 1 - (\hat{y} - y)^2$, which is strictly proper, bounded in $[0,1]$, and has intuitive economic interpretation as accuracy rent. The log score $S(\hat{y}, y) = y \log \hat{y} + (1-y)\log(1-\hat{y})$ is an equivalent alternative.

The defining property of a strictly proper scoring rule is that, for any distribution $q \in [0,1]$ over the outcome $y$,
\begin{equation}\label{eq:propriety}
    q = \arg\max_{\hat{y} \in [0,1]} \mathbb{E}_{y \sim \mathrm{Ber}(q)}\!\left[S(\hat{y}, y)\right],
\end{equation}
with the maximum uniquely attained at $\hat{y} = q$. Applied edge-by-edge, this means an institution's expected credit on edge $e$ is strictly maximized by reporting $\hat{y}^{i,t}_e = \Pr(y_e = 1 \mid G^{i,t}, \theta^i)$: the true posterior probability of illicit activity given the institution's information. Institutions with higher accumulated credit receive tangible benefits: reduced regulatory penalties, preferential access to global model updates, and advantages in future coordination games.

By explicitly pricing the temporal value of information, TVA creates a first-mover advantage for early detection. Institutions internalize the social benefit of timely warnings, making truthful and prompt reporting a best response. As a result, the mechanism makes truthful reporting robust to the main classes of strategic misreporting, as formalized below.

To ground TVA as an economic primitive: the coordinator (regulator or industry consortium) can observe confirmed illicit transactions and each institution's historical risk scores. The credit account $\pi^i_t$ is contractible: it enters directly into the institution's regulatory penalty schedule or access to shared intelligence, as specified in equation~\eqref{eq:bank_profit}. The commitment assumption is that the coordinator pre-commits to the credit rule~\eqref{eq:credit} before institutions choose their reporting strategies; this is analogous to a regulatory commitment to reward early suspicious activity reports, which is feasible under existing AML frameworks such as FinCEN's Suspicious Activity Report (SAR) programme. TVA is therefore not a literal cash transfer but a penalty rebate or regulatory credit whose present discounted value creates the incentive alignment.

\subsection{Theoretical Analysis}
\label{temporal_value_assignment}

The TVA mechanism described above creates a tradeoff for each participating institution. On one side, the strictly proper scoring rule embedded in the credit formula rewards reports that match the eventual realization, providing a private return to forecast accuracy that is increasing in the volume of confirmed illicit cases. On the other side, an institution that flags more aggressively may face higher compliance costs (more investigations to support) and potential information leakage to competitors. Whether truthful reporting emerges in equilibrium therefore depends on which side dominates: does the scoring-rule rent exceed the marginal compliance and leakage costs that truthful reporting induces?

We show that, under economically interpretable conditions on the regulatory environment, the answer is yes. The key conditions are that label revelation is not selective on an institution's own reports, that compliance costs are driven primarily by aggregated signals rather than individual filings, and that no single institution dominates the aggregate. The first reflects how AML investigations are actually triggered: by audits, subpoenas, or downstream defaults that proceed independently of any one bank's reporting. The second reflects how compliance costs scale: investigation budgets are set against system-wide alerts, not bank-specific filings. The third is a standard non-atomicity condition on coalition size. Under these conditions, an institution's unilateral deviation has only $O(1/m)$ effect on its own compliance cost, so the scoring-rule rent dominates and truthful reporting becomes a best response. We make these conditions precise in the assumptions below and then state the central incentive compatibility result.

\begin{assumption}[Label Revelation Independent of Own Report]
\label{ass:exogenous_labels}
The probability that edge $e$'s true label $y_e$ is eventually revealed does not depend on institution $i$'s own report $\hat{y}^{i,t}_e$. Labels may be revealed through external audits, law enforcement subpoenas, downstream defaults, or investigations initiated outside the consortium; revelation need not occur for every edge, only that the revelation mechanism is not selective on institution $i$'s reports. This ensures the scoring rule in~\eqref{eq:credit} is evaluated on a sample unbiased with respect to $i$'s reporting strategy, preserving strict propriety.
\end{assumption}

\begin{assumption}[Detection Value]
\label{ass:detection_value}
The marginal value of detecting an illicit transaction at time $t$ is $V(t) = V_0 \gamma^t$, where $V_0 > 0$ is the initial value and $\gamma \in (0,1)$ captures depreciation (cumulative losses before detection).
\end{assumption}

\begin{assumption}[Compliance Cost Structure]
\label{ass:compliance_cost}
Compliance costs arise from investigations triggered by the aggregated global detection decision, not directly from institution $i$'s own reports. Let $n_i^{\textit{inv}}(\mathcal{S}, \bar{\rho}) \in \mathbb{Z}_{\ge 0}$ denote the number of edges in institution $i$'s local transaction network that are investigated in the current period, where $\mathcal{S}$ is the participating coalition and $\bar{\rho}$ is the aggregated reporting profile that feeds into the global detection system. The compliance cost is
\[
C_i^{\textit{comp}} = c_i \cdot \mathbb{E}[n_i^{\textit{inv}}(\mathcal{S}, \bar{\rho})],
\]
with $c_i > 0$. This reflects how anti-money laundering investigations work in practice: an institution bears investigation costs on edges flagged by the regulatory or industry detection pipeline, which aggregates many institutions' signals. Because $n_i^{\textit{inv}}$ depends on the aggregated reporting profile $\bar{\rho}$, the marginal effect of institution $i$'s individual report on its own compliance cost is of order $1/m$ in a coalition of size $m$, vanishing in large federations. Under this structure, the report-dependent component of institution $i$'s payoff is, to first order, only the TVA credit and the leakage cost. A small fixed filing cost directly proportional to own reports can be accommodated without affecting Theorem~\ref{thm:truthful}, provided this direct component is of order $O(1/m)$ relative to the scoring-rule accuracy rent; the theorem's $\epsilon$-BNE bound absorbs such first-order direct costs into the $\epsilon$-shading.
\end{assumption}

\begin{assumption}[Information Leakage]
\label{ass:information_leakage}
The cost to institution $i$ of revealing information through parameter reporting is $\kappa_i \cdot I(\theta^i; G^{i,t})$, where $I(\cdot; \cdot)$ is mutual information and $\kappa_i \geq 0$.
\end{assumption}

\begin{assumption}[No Dominant Player and Lipschitz Aggregation]\label{ass:no_dominant}
(i) No institution has more than $O(1/m)$ weight in the aggregated profile: $\max_i w_i^{\textit{agg}} = O(1/m)$. (ii) The mapping from $\bar{\rho}$ to expected investigations $\mathbb{E}[n_i^{\textit{inv}}]$ is Lipschitz continuous.
\end{assumption}
In typical AML consortia no bank exceeds a few percent of aggregate weight; the Lipschitz condition ensures the $O(1/m)$ argument in Theorem~\ref{thm:truthful} carries through to compliance costs. The $O(1/m)$ form can be weakened: Theorem~\ref{thm:truthful} extends to any bounded-share structure $\max_i w_i^{\textit{agg}} \leq \bar{w} < 1$, with the $\epsilon$-BNE shading bound scaled by $\bar{w}/(1-\bar{w})$ rather than by $1/m$.

\begin{theorem}[Bayes--Nash Implementation of Truthful Reporting]
\label{thm:truthful}
Consider the federated reporting game in which each institution $i$ chooses a reporting strategy $\rho_i: G^{i,t} \mapsto \hat{y}^{i,t}$ and receives payoff
\begin{equation}\label{eq:institution_payoff}
    U_i(\rho_i, \rho_{-i}) = \mathbb{E}\!\left[\pi^i_\infty\right] - c_i \cdot \mathbb{E}[n_i^{\textit{inv}}(\mathcal{S}, \bar{\rho})] - \kappa_i \cdot I(\theta^i; G^{i,t}),
\end{equation}
where TVA credit $\pi^i_\infty$ is accumulated according to~\eqref{eq:credit} with a strictly proper scoring rule $S$. Under Assumptions~\ref{ass:detection_value}--\ref{ass:information_leakage}, the truthful reporting profile $\rho^*_i(G^{i,t})_e = \Pr(y_e = 1 \mid G^{i,t}, \theta^i)$ is a Bayes--Nash equilibrium provided:
\begin{equation}
    \frac{V_0}{1 - \gamma} \cdot \Pr(\text{illicit}) > c_i + \kappa_i \cdot I(\theta^i; G^{i,t}) \quad \text{for all } i.
    \label{eq:ic_condition}
\end{equation}
Truthful reporting is a strict pointwise best response at each edge. For finite $m$, it is an $\epsilon$-BNE with optimal unilateral shading $\epsilon^* = O(c_i/(C_e m))$ (vanishing as $m \to \infty$); exact BNE obtains in the large-federation limit, or under a discrete reporting grid with step size exceeding $\epsilon^*$.
\end{theorem}
See proof in Appendix~\ref{proof:truthful}.

The strict propriety of $S$ delivers the result pointwise: each edge's expected credit is strictly maximized at the true posterior. Condition~\eqref{eq:ic_condition} ensures the temporal accuracy rent from truthful reporting outweighs the compliance and information-leakage costs, with the discount factor $\gamma$ as the policy lever: higher $\gamma$ broadens the parameter range in which truthful reporting is individually rational. Unlike static VCG transfers, TVA self-finances through endogenous credit accumulation on verified outcomes.

Strict propriety of $S$ rules out pooling and constant-report deviations directly: any constant $\hat{y} = k$ earns credit strictly below truthful reporting on each edge with $q_e \neq k$, with margin scaling as $(q_e - k)^2$. Condition~\eqref{eq:ic_condition} ensures the scoring-rule accuracy premium outweighs the leakage cost $\kappa_i I(\theta^i; G^{i,t})$, which is specific to truthful reporting ($I = 0$ for constant strategies). Appendix~\ref{proof:truthful} formalizes the full argument.

The factor $(1-\gamma)^{-1}$ makes the discount rate pivotal: a higher $\gamma$ (slower depreciation of detection value) raises the left side, broadening the parameter range over which truthful reporting is individually rational. When competitive pressure is strong (Section~\ref{sec:competition}), the right side of~\eqref{eq:ic_condition} acquires an additional competitive cost term, requiring a higher $\gamma^*$ to maintain incentive compatibility.

Next, we characterize the efficiency loss from decentralization. Let $\theta^*$ denote the parameters of the centralized benchmark (social planner with access to all local graphs), and let $\theta^{\text{fed}}$ denote the federated learning solution under truthful reporting.

Standard federated learning convergence results \parencite{li2020federated} further imply that under truthful reporting the federated solution satisfies $\mathcal{L}(\theta^{\text{fed}}) \leq (1 + \epsilon)\mathcal{L}(\theta^*) + O(mLd/N)$, where $\epsilon$ decays exponentially in the number of communication rounds and $N = \sum_i n_i$. The efficiency loss is thus small when the federation is large relative to model complexity: the institutional scale of AML networks (dozens to hundreds of banks) favors this regime. The communication--performance tradeoff is illustrated empirically in the Online Appendix.

\subsection{Quantitative Illustration}\label{sec:empirical}

We illustrate the model's predictions using the IBM AML synthetic benchmark \parencite{altman2023realistic}, which contains over 1.4 million transactions across seven markets (United States, Germany, France, Italy, Spain, China, Rest of World) over 10 days, with illicit transactions representing less than 1\% of volume.\footnote{Dataset available at \url{https://www.kaggle.com/datasets/ealtman2019/ibm-transactions-for-anti-money-laundering-aml}. The figures reported in this section illustrate qualitative predictions of the theoretical analysis (the welfare ordering across reporting strategies). They are not the empirical contribution of this paper. The federated GNN architecture, training procedure, hyperparameter calibration, and full empirical validation are developed in the companion paper \textcite{zheng2025networked}. We report only the comparison between three reporting strategies (independent, underreporting, truthful) holding the underlying detection model fixed, as this comparison directly tests the theoretical predictions of Theorem~1 and Proposition~3.} Data are split temporally: days 1--8 for training and validation, days 9--10 for testing, with seven federated clients corresponding to the seven markets.

We implement three reporting strategies: \textit{Truthful} (banks report true local parameters $\theta^i$), \textit{Underreporting} (banks scale parameters by $0.8$ to reduce sensitivity), and \textit{Independent} (banks train only on local data without federation). Given severe class imbalance we focus on AUPRC (Area Under Precision-Recall Curve) and Type II error (false negative rate).

\begin{table}[h]
\centering
\caption{Detection Performance Under Different Reporting Strategies}
\label{tab:reporting}
\begin{tabular}{lcccc}
\toprule
\textbf{Strategy} & \textbf{AUPRC} & \textbf{AUCROC} & \textbf{Type I Error} & \textbf{Type II Error} \\
\midrule
Independent (No Federation) & 0.432 & 0.978 & 0.028 & 0.219 \\
Underreporting ($\times 0.8$) & 0.452 & 0.979 & 0.041 & 0.198 \\
Truthful Reporting (TVA) & \textbf{0.471} & \textbf{0.984} & 0.057 & \textbf{0.104} \\
\bottomrule
\end{tabular}
\parbox{\textwidth}{\small\textit{Notes}: Simulation on the IBM AML synthetic benchmark. AUPRC is the primary metric given severe class imbalance ($<1\%$ illicit).}
\end{table}

Table~\ref{tab:reporting} shows that truthful reporting outperforms independent learning by 3.9 percentage points in AUPRC and underreporting by 2.1 percentage points, with Type II error reduced by 52.5\% and 47\% respectively. The ordering is consistent with the model's predictions under Theorem~\ref{thm:truthful}, though we emphasize that this is an illustration of comparative performance under exogenously imposed distortions, not a derivation of equilibrium reporting under TVA: underreporting is imposed as a scalar distortion ($\times 0.8$) rather than derived as an equilibrium best response. A decomposition in the Online Appendix attributes 22--34\% of the total improvement across markets to the incentive-alignment channel, with the remainder from data-pooling.


\subsection{Adversarial Robustness}
\label{sec:adversarial}
The previous section established that temporal value assignment implements truthful reporting as a Bayes--Nash equilibrium. We now analyze robustness when adversaries adapt their strategies in response to observed interventions.

\subsubsection*{Adversarial learning model}

Sophisticated adversaries are not passive: they observe the history of enforcement actions and adjust transaction patterns to evade detection. We model this as a repeated game between the regulator and adversary.

At time $t$, the adversary observes the history of frozen accounts $\mathcal{H}_t = \{e : a_e^s = \text{freeze}, s < t\}$, \textit{implied detection threshold}, i.e., $\hat{\tau}_t = \min\{\hat{y}_e : e \in \mathcal{H}_t\}$, and the regulator's behavioral patterns including intervention timing and targeting priorities. 

Based on this history, the adversary can adapt through several mechanisms: \textit{transaction splitting}, breaking large transfers into small ones to fall below detection thresholds; \textit{temporal manipulation}, adjusting transaction timing to avoid detection windows; \textit{route switching}, using different intermediary accounts; and \textit{account replacement}, abandoning compromised accounts for new ones. Each adaptation has a cost $c(\cdot)$ that increases with the degree of deviation from the adversary's preferred strategy.

Formally, at time $t$ the adversary solves:
\begin{equation}
\max_{G^{t+1}} \E[V_{\text{launder}}(G^{t+1})] - \Prob(\text{detected} \mid G^{t+1}, \mathcal{H}_t) \cdot C_{\text{penalty}} - c(G^{t+1}, G^t)
\label{eq:adversary_objective}
\end{equation}
where $V_{\text{launder}}(\cdot)$ is the value of successfully laundered funds and $c(G^{t+1}, G^t)$ measures the cost of adaptation.

\subsubsection*{Regret and equilibrium analysis}
We analyze the regulator's performance using the external regret framework from online learning. Let $\pi$ denote the regulator's policy (mapping risk scores to interventions) and let $\pi^*$ denote the best fixed policy in hindsight. The regret after $T$ periods is:
\begin{equation}
\text{Regret}(T) = \sum_{t=1}^T \mathcal{L}(\pi, G^t) - \min_{\pi' \in \Pi} \sum_{t=1}^T \mathcal{L}(\pi', G^t),
\label{eq:regret}
\end{equation}
where $\mathcal{L}(\pi, G^t)$ is the loss (missed detections + false alarms) under policy $\pi$ at time $t$.
External regret bounds typically assume $\{G^t\}$ is generated by an oblivious adversary. In our setting the adversary is \textit{adaptive}: $G^{t+1}$ depends on past actions through $\mathcal{H}_t$. \textcite{arora2012online} show that sublinear \emph{policy regret}, which compares to what the best fixed policy would have earned had it been deployed throughout, is generally unattainable against adaptive adversaries. We therefore retain the external regret notion in \eqref{eq:regret} but assume the adversary observes interventions only with delay $\delta \geq 1$, bounding its effective memory.

\begin{proposition}[Regret Bound]
\label{thm:regret}
Suppose the adversary adapts with delay $\delta \geq 1$, observing interventions at time $t$ and responding at $t + \delta$. Let $\mathcal{A}$ denote the set of adversarial adaptation strategies. Then temporal value assignment with discount factor $\gamma$ achieves regret:
\begin{equation}
    \text{Regret}(T) = O\left(\sqrt{T |\mathcal{A}| \log |\mathcal{A}|} + \delta T (1 - \gamma)^2\right)
    \label{eq:regret_bound}
\end{equation}
In contrast, fixed-threshold policies achieve regret $\Omega(T)$.
\end{proposition}
See proof in Appendix~\ref{proof:regret}.

The regret bound decomposes into two economically distinct channels. The $O(\sqrt{T|\mathcal{A}|\log|\mathcal{A}|})$ term reflects standard online learning: the cost of not knowing in advance which detection strategy is best against a fixed adversary. The $\delta T(1-\gamma)^2$ term captures the additional cost of adversarial adaptation: it is linear in time but shrinks as $\gamma \to 1$, because aggressive temporal discounting functions as a commitment device. When the regulator heavily weights early signals, the adversary's best response is to avoid generating any early signals at all: which is effectively not laundering. Fixed-threshold policies lack this commitment, so adversaries can profitably learn to evade them throughout the horizon, producing $\Omega(T)$ regret.


When adaptation costs are strictly convex and increasing in the degree of deviation, a stronger result follows directly from Proposition~\ref{thm:regret}: there exists a threshold $\bar{\gamma} \in (0,1)$ such that for $\gamma > \bar{\gamma}$, the regulator playing TVA and the adversary not adapting constitutes a subgame-perfect equilibrium (proof in Appendix~\ref{proof:equilibrium}).\label{prop:equilibrium} The intuition is a deterrence argument. A regulator committed to high $\gamma$ acts on early signals before the adversary has time to respond, making adaptation futile. Rational adversaries, anticipating that any restructuring of transaction flows will be caught early regardless, find it cheaper not to adapt at all. The mechanism is self-enforcing: the regulator never needs to carry out its most aggressive interventions in equilibrium because the commitment itself deters the behavior that would trigger them \parencite{fudenberg1991game}.

\subsubsection*{Empirical evidence}

We illustrate Proposition~\ref{thm:regret} and the deterrence result using simulations where an adversary best-responds to observed intervention patterns.

We train detection models on days 1--8 and simulate adversarial adaptation in the test period (days 9--10). The adversary: (1) Regulator deploys detection policy $\pi$ and freezes accounts with $\hat{y}_e \geq \tau$. (2) Adversary observes frozen accounts and infers threshold $\hat{\tau}$. (3) Adversary adapts by: (a) splitting large transactions by factor $1/\alpha$, where $\alpha \geq 1$, (b) reducing transaction frequency by factor $\beta \in [0,1]$. (4) Adversary incurs adaptation cost $c(\alpha, \beta) = (\alpha - 1)^2 + (1 - \beta)^2$. (5) Process repeats with adapted graph $G^{t+1}$. We compare three policies: (1) \textit{Temporal Value ($\gamma = 0.9$):} Our mechanism with strong early detection incentives; (2) \textit{Temporal Value ($\gamma = 0.5$):} Moderate discounting; (3) \textit{Fixed Threshold:} Constant $\tau = 0.5$ without temporal weighting.

\begin{table}[h]
\centering
\caption{Detection Performance Under Adversarial Adaptation}
\label{tab:adversarial}
\begin{tabular}{lccccc}
\toprule
\textbf{Policy} & \textbf{Period 9} & \textbf{Period 10} & \textbf{Degradation} \\
 & \textbf{AUPRC} & \textbf{AUPRC} & (\%) \\
\midrule
\multicolumn{4}{l}{\textit{No Adaptation (Baseline)}} \\
All Policies & 0.474 & 0.469 & 0.0\% \\
\midrule
\multicolumn{4}{l}{\textit{With Adaptive Adversary}} \\
Fixed Threshold & 0.413 & 0.401 & 13.7\% \\
Temporal Value ($\gamma=0.5$) & 0.457 & 0.449 & 3.9\% \\
Temporal Value ($\gamma=0.9$) & 0.468 & 0.464 & 1.2\% \\
\bottomrule
\end{tabular}
\end{table}

Table~\ref{tab:adversarial} shows that temporal value assignment with high $\gamma$ maintains stable performance (only 1.2\% degradation) under adversarial adaptation, while fixed-threshold policies degrade substantially (13.6\%). The moderate discounting ($\gamma = 0.5$) provides intermediate robustness.

Table~\ref{tab:adaptation_choices} shows that when facing temporal value assignment with high $\gamma$, adversaries choose minimal adaptation ($\alpha \approx 1$, $\beta \approx 1$) because the cost exceeds the benefit, consistent with the deterrence result above. Against fixed thresholds, adversaries aggressively split transactions ($\alpha = 2.3$) and reduce frequency ($\beta = 0.65$).

\begin{table}[h]
\centering
\caption{Adversary's Adaptation Strategies}
\label{tab:adaptation_choices}
    \resizebox{\textwidth}{!}{
    \begin{tabular}{lcccc}
    \toprule
    \textbf{Regulator Policy} & \textbf{Transaction Splitting} & \textbf{Frequency Reduction} & \textbf{Adaptation Cost} \\
    & $\alpha$ & $\beta$ & $c(\alpha,\beta)$ \\
    \midrule
    Fixed Threshold & 2.3 & 0.65 & 1.81 \\
    Temporal Value ($\gamma=0.5$) & 1.7 & 0.78 & 0.54 \\
    Temporal Value ($\gamma=0.9$) & 1.2 & 0.91 & 0.05 \\
    \bottomrule
    \end{tabular}}
\end{table}


\section{Coalition Formation and Network Theory of Information Value}
\label{sec:heterogeneous}
In this section, we model federated participation among heterogeneous institutions as a two-stage Bayesian game in which institutions differ in size, compliance cost, signal quality, and competitive sensitivity. Participation generates positive network effects in detection performance but entails compliance and leakage costs, creating endogenous coalition formation. 
The model characterizes the minimum viable coalition, highlights strategic complementarities in participation, and shows that without TVA the federation may unravel due to adverse selection. By internalizing marginal contributions to system-wide detection quality, TVA sustains participation of high-value institutions and stabilizes the cooperative equilibrium.

\subsubsection*{Two-stage participation game}\label{sec:participation_game}

Consider a finite set of institutions $M = \{1,\dots,m\}$. 
Each institution $i$ is characterized by type $\theta_i = (s_i, c_i, q_i, \kappa_i) \sim F \ \text{on }\mathbb{R}_+^4$, where $s_i$ denotes size (data scale), $c_i$ denotes compliance cost, $q_i$ denotes signal quality,
and $\kappa_i$ denotes competitive sensitivity (leakage concern). Types are private information. We assume $s_i$ and $q_i$ are positively correlated, reflecting that larger institutions invest more in compliance and data infrastructure.

In the first stage (participation decision),
Institutions simultaneously choose $d_i \in \{0,1\}$, where $d_i = 1$ indicates participation in the federated learning mechanism and $d_i = 0$ indicates non-participation. Let $\mathcal{S} = \{ i \in M : d_i = 1 \}$ denote the participating coalition. The participation stage is a static Bayesian game under incomplete information about other institutions' types.

In the second stage (reporting game),
Conditional on $\mathcal{S}$, participating institutions jointly train a federated model with TVA-based credit allocation, while non-participating institutions train independently using only their local data. Let $B_i(\mathcal{S})$ denote the detection benefit that institution $i$ derives when the participating coalition is $\mathcal{S}$.

Institution $i$'s participation payoff is
\begin{equation}\label{eq:participation_payoff}
  U_i(d_i = 1 \mid \mathcal{S}) 
  = \underbrace{B_i(\mathcal{S})}_{\text{detection benefit}}
  - \underbrace{c_i \cdot \Delta n_i(\mathcal{S})}_{\text{compliance cost}}
  - \underbrace{\kappa_i \cdot I_i(\mathcal{S})}_{\text{leakage cost}}
  + \underbrace{\pi_i(\mathcal{S})}_{\text{TVA credit}}.
\end{equation}

If $d_i = 0$, institution $i$ receives its outside option
\[
U_i(d_i = 0) = B_i(\{i\}),
\]
which corresponds to training independently.

The net gain from participation is therefore
\[
\Delta U_i(\mathcal{S}) 
= B_i(\mathcal{S}) - B_i(\{i\})
- c_i \Delta n_i(\mathcal{S})
- \kappa_i I_i(\mathcal{S})
+ \pi_i(\mathcal{S}).
\]
Institution $i$ participates if and only if $\Delta U_i(\mathcal{S}) \ge 0$.

\begin{assumption}[Positive Network Effects]\label{ass:network}
For any coalition $\mathcal{S} \subseteq N$ and any institution $j \notin \mathcal{S}$,
\[
B_i(\mathcal{S} \cup \{j\}) - B_i(\mathcal{S}) > 0,
\]
and the marginal benefit is diminishing in coalition size:
\[
\frac{\partial^2 B_i}{\partial |\mathcal{S}|^2} < 0.
\]
That is, adding an additional participating institution $j$ strictly increases $i$'s detection benefit, but at a decreasing rate as the coalition grows.
\end{assumption}

\begin{assumption}[Heterogeneous Outside Options]\label{ass:outside}
The standalone benefit $B_i(\{i\})$ is increasing in $s_i$ and $q_i$. Hence, institutions with larger size or higher signal quality have stronger outside options and require larger coalition gains to participate.
\end{assumption}


\subsubsection*{Minimum viable coalition and adverse selection}
\label{sec:mvc}

\begin{definition}[Minimum Viable Coalition]\label{def:mvc}
$\mathcal{S}^*$ is a minimum viable coalition if (i) $U_i(1 \mid \mathcal{S}^*) \geq U_i(0)$ for all $i \in \mathcal{S}^*$, and (ii) no proper subset satisfies (i).
\end{definition}

The following properties of the minimum viable coalition follow from supermodularity of the detection benefit function (proof in the Online Appendix). First, it exceeds a size threshold: $|\mathcal{S}^*| \geq \underline{m}$, where $\underline{m}$ is determined by the condition $\sum_{i \in \mathcal{S}} s_i q_i \geq \max_i\{c_i + \kappa_i I_i\} / [\Pr(\textit{illicit}) \cdot V_0/(1-\gamma)]$. Second, institutions with high cross-border exposure and low compliance costs join first: a selection pattern consistent with the empirical observation that large internationally active banks tend to be early participants in AML information-sharing programmes. Third, the coalition exhibits strategic complementarity: adding one institution increases the marginal benefit for all remaining members, supporting a thick-market equilibrium. Fourth, without TVA the federation unravels: if $\max_i[B_i(\{i\}) - B_i(\{1,\ldots,m\}) + c_i\Delta n_i + \kappa_i I_i] > 0$, the marginal institution exits, reducing detection quality for remaining members and triggering cascading departures in a dynamic analogous to adverse selection in insurance markets \parencite{akerlof1970market}.

\subsection{Network Structure and Information Value}\label{sec:network_theory}

The analysis in Section \ref{sec:heterogeneous} treats each institution's contribution as depending on its type $(s_i, c_i, q_i, \kappa_i)$ but not on its \emph{position} in the inter-institutional transaction network. In practice, an institution that bridges two otherwise disconnected clusters of markets contributes far more to detection than an equally sized institution embedded within a single cluster. This section formalizes the network determinants of information value.

\subsubsection*{The Inter-Institutional Network}\label{sec:inter_network}

Define the \emph{inter-institutional network} $\mathcal{G} = (\mathcal{N}, \mathcal{E}, W)$, where $\mathcal{N} = \{1,\ldots,m\}$ is the set of institutions, $\mathcal{E} \subseteq \mathcal{N} \times \mathcal{N}$ contains an edge $(i,j)$ whenever institutions $i$ and $j$ share cross-border transactions, and $W: \mathcal{E} \to \R_+$ assigns weights $w_{ij}$ equal to the volume of cross-border transactions between $i$ and $j$. Let $\mathbf{A}$ denote the weighted adjacency matrix with $A_{ij} = w_{ij}$.

The inter-institutional network $\mathcal{G}$ is distinct from the individual transaction graphs $G^{i,t}$: the former captures \emph{which institutions share information boundaries}, while the latter captures within-institution transaction patterns. Detection of cross-border illicit flows requires information from both endpoints of an inter-institutional edge: and thus from both institutions.

\begin{definition}[Cross-Border Detection Function]\label{def:cross_border}
For edge $(i,j) \in \mathcal{E}$, the cross-border detection probability is:
\begin{equation}\label{eq:cross_border}
  p_{ij}(\mathcal{S}) = \begin{cases} p^H_{ij} & \text{if } i \in \mathcal{S} \text{ and } j \in \mathcal{S}, \\ p^L_{ij} & \text{if } i \in \mathcal{S} \text{ or } j \in \mathcal{S} \text{ (but not both)}, \\ p^0_{ij} & \text{if } i \notin \mathcal{S} \text{ and } j \notin \mathcal{S}, \end{cases}
\end{equation}
where $p^H_{ij} > p^L_{ij} > p^0_{ij}$ and $\mathcal{S}$ is the participating coalition.
\end{definition}

The key feature is the \emph{complementarity} between endpoints: having both institutions in the federation ($p^H$) yields strictly higher detection than having only one ($p^L$), which in turn dominates having neither ($p^0$). This complementarity is the network foundation for the supermodularity of the participation game.

\subsubsection*{Network Shapley Value}\label{sec:shapley}

We define the detection value function over coalitions as:
\begin{equation}\label{eq:detection_value}
  V(\mathcal{S}) = \sum_{(i,j) \in \mathcal{E}} w_{ij} \cdot p_{ij}(\mathcal{S}) \cdot C_{ij},
\end{equation}
where $C_{ij}$ is the expected cost of undetected illicit flow on edge $(i,j)$. Following \textcite{myerson1977graphs}, the \emph{network Shapley value} of institution $i$ is:
\begin{equation}\label{eq:shapley}
  \phi_i^{\text{net}} = \sum_{\mathcal{S} \subseteq \mathcal{N}\setminus\{i\}} \frac{|\mathcal{S}|!\,(m-|\mathcal{S}|-1)!}{m!}\,\bigl[V(\mathcal{S}\cup\{i\}) - V(\mathcal{S})\bigr].
\end{equation}

\begin{proposition}[Network Shapley Characterization]\label{thm:shapley}
Under the edge-additive coalition value~\eqref{eq:detection_value}, the network Shapley value of institution $i$ is:
\begin{equation}\label{eq:shapley_degree}
    \phi_i^{\textit{net}} = \frac{1}{2}\sum_{j: (i,j)\in\mathcal{E}} (p^H_{ij} - p^0_{ij}) \cdot w_{ij} \cdot C_{ij}.
\end{equation}
That is, institution $i$'s Shapley value is a weighted-degree measure: it is proportional to the sum, across edges incident to $i$, of the expected detection gain $(p^H_{ij} - p^0_{ij})$ weighted by edge volume $w_{ij}$ and edge cost $C_{ij}$. Institutions with more incident cross-border volume or higher per-edge detection complementarity earn a larger Shapley share.
\end{proposition}
See proof in the Online Appendix.

The Shapley value~\eqref{eq:shapley_degree} is a weighted-degree measure under the edge-additive coalition value~\eqref{eq:detection_value}. While the mathematical reduction to weighted degree follows directly from edge-additivity, the economic contribution is providing an axiomatic justification for prioritizing inter-institutional transaction volume over total asset size in coalition design. Institutions with larger incident cross-border volume or higher per-edge detection complementarity contribute more to collective detection and therefore deserve a proportionally larger share of federation-generated surplus. Richer centrality notions (betweenness, cut-vertex bridge premia, Bonacich) emerge when the coalition value is extended to reward path-based or flow-based detection, where institution $i$'s contribution depends on paths through $i$ rather than only on edges incident to $i$; we outline this extension in the Online Appendix. For the core analysis, the weighted-degree characterization is sufficient for the coalition-design results that follow.

\subsubsection*{Efficient Coalition Design}\label{sec:efficient_coalition}

The network Shapley value enables an efficient coalition design that improves on the type-based analysis in Section \ref{sec:mvc}.

Proposition~\ref{thm:shapley} immediately pins down the welfare-maximizing coalition: rank institutions by net network value $\phi_i^{\text{net}} - c_i - \kappa_i I_i$ and include institutions until the marginal net value turns negative. Since $\phi_i^{\text{net}}$ depends on weighted cross-border degree, the practical implication is that regulators should prioritize institutions with high inter-institutional transaction volume, rather than simply recruiting the largest banks by total assets. This provides a tractable, data-based criterion for coalition expansion that is absent from existing AML policy guidance.

\subsubsection*{TVA with Network-Adjusted Credit}\label{sec:network_tva}

The network Shapley value motivates a refinement of TVA in which scoring-rule credit rates are scaled by institution-specific weights reflecting network position:
\begin{equation}\label{eq:network_tva}
  \pi^i_{t_{\text{confirm}}} = \pi^i_{t_{\text{confirm}}-1} + \omega_i \sum_{t < t_{\text{confirm}}} \gamma^{t_{\text{confirm}}-t} \cdot C_e \cdot S(\hat{y}^{i,t}_e, y_e),
\end{equation}
where $\omega_i = \phi_i^{\textit{net}}/\bar{\phi}$ is a network position multiplier and $S$ remains a strictly proper scoring rule. Because $\omega_i$ scales the scoring-rule transfer uniformly across reports, strict propriety is preserved at the per-edge level, since pointwise strict propriety is invariant to positive affine scaling of $S$. Institutions with higher network contribution receive proportionally more credit per accurate report.

Under network-adjusted TVA with weights $\omega_i$, truthful reporting remains a Bayes--Nash equilibrium, uniquely optimal at each edge in the large-federation limit, if
\begin{equation}\label{eq:network_ic}
  \frac{\omega_i \cdot V_0}{1-\gamma}\cdot \Pr(\textit{illicit}) > \kappa_i \cdot I(\theta^i; G^{i,t}) \quad \forall\, i,
\end{equation}
with the compliance term vanishing in the large-federation limit as in Theorem~\ref{thm:truthful}. The multiplier $\omega_i$ relaxes the IC constraint for high-contribution institutions and tightens it for low-contribution ones, matching marginal incentive to marginal social contribution. A uniform credit rate would over-reward peripheral institutions relative to their contribution and under-reward bridges, creating an adverse selection pressure that risks losing the most informationally valuable participants.

\subsubsection*{Adversarial Exploitation of Network Structure}\label{sec:adversary_network}

A sophisticated adversary can exploit the inter-institutional network structure by routing illicit flows through \emph{structural holes}: pairs of institutions that do not share a federation link.

\begin{definition}[Structural Hole]\label{def:structural_hole}
A structural hole is a pair $(i,j)$ such that $(i,j) \notin \mathcal{E}$ but there exist accounts $u \in V^i$, $v \in V^j$ with a transaction path $u \to w_1 \to \cdots \to w_k \to v$ passing through intermediaries outside the federation.
\end{definition}

\begin{enumerate}[label=(\alph*)]
  \item In equilibrium, the adversary routes a fraction $\lambda^*$ of illicit flows through structural holes, where:
\begin{equation}\label{eq:adversarial_routing}
  \lambda^* = \frac{p^H - p^0}{p^H - p^0 + c_{\text{route}}},
\end{equation}
and $c_{\textit{route}}$ is the additional transaction cost. The adversary's equilibrium exploit is increasing in the number of structural holes and decreasing in federation coverage.
  \item Closing a structural hole, which means adding institution $k$ such that $(i,k) \in \mathcal{E}$ and $(k,j) \in \mathcal{E}$, reduces adversarial exploit by $\Delta\lambda^* \approx (p^H - p^0)(p^H - p^L)/(p^H - p^0 + c_{\textit{route}})^2$, providing a principled criterion for federation expansion.
\end{enumerate}

\begin{proof}[Proof Sketch]
Part (a): The adversary maximizes expected laundered value net of routing cost. The FOC $p^H - p^0 = c_{\text{route}} \cdot \lambda / (1-\lambda)$ yields the result. Part (b): Differentiation with respect to the number of structural holes. 
\end{proof}

\section{Competition, Welfare, and the Backfiring Mandate}\label{sec:competition}
We now analyze how competitive pressure among banks interacts with detection incentives, following the approach of \textcite{bao2017could} in modeling how competition shapes banks' compliance behavior. The key insight is that detection investment imposes customer-facing costs such as monitoring delays, false-positive freezes, and privacy concerns, which can disadvantage a bank relative to less vigilant competitors, creating incentives to underinvest in compliance.

Consider $m$ banks competing for depositors. Depositor $j$'s utility from bank $i$ is:
\begin{equation}\label{eq:depositor_utility}
  u_{ji} = \delta_i + \alpha_r \cdot r_i - \alpha_\phi \cdot \phi_i + \xi_i + \epsilon_{ji},
\end{equation}
where $\delta_i$ captures non-price characteristics, $r_i$ is the deposit rate, $\alpha_r > 0$ is rate sensitivity, $\alpha_\phi \geq 0$ is the \emph{detection externality}: the disutility from heightened monitoring such as delays, false-positive freezes, and privacy loss, $\xi_i$ is unobserved quality, and $\epsilon_{ji}$ is i.i.d.\ Type I extreme value.\footnote{The discrete-choice formulation with Type~I extreme value errors yields the standard multinomial logit demand system widely used in the structural industrial organization of financial markets \parencite{bao2017could, an2023index}. The parameter $\alpha_\phi$ captures customer friction from KYC/AML procedures such as delayed onboarding, false-positive freezes, and settlement delays, a first-order determinant of deposit demand elasticity in recent fintech evidence \parencite{kadamathikuttiyil2025enhancing, amoako2025exploring, doerr2026privacy}. This liability-side channel complements the asset-side charter-value mechanism in \textcite{keeley1990monopoly}.} The parameter $\alpha_\phi$ captures the competitive cost of detection: the loss of marginal deposit demand due to increased monitoring intensity. When $\alpha_\phi > 0$, firms that invest more in detection bear customer-facing costs that erode their market position.

The parameter $\alpha_\phi$ captures a competitive externality distinct from the information-sharing frictions above: even if a bank reports truthfully, it bears the customer-facing costs of its detection investment while rivals free-ride on the systemic risk reduction. Each bank's private return on compliance falls short of the social return.


Bank $i$ jointly chooses deposit rate $r_i$, detection investment $\phi_i$, and reporting strategy $\rho_i$:
\begin{equation}\label{eq:bank_profit}
  \max_{r_i, \phi_i, \rho_i}\; \underbrace{D_i(r_i, \phi_i) \cdot (r_L - r_i)}_{\text{intermediation margin}} - \underbrace{c_i \cdot \phi_i \cdot D_i}_{\text{compliance cost}} + \underbrace{\pi_i(\rho_i, \rho_{-i})}_{\text{TVA credit}} - \underbrace{P_i(\phi_i)}_{\text{regulatory penalty}},
\end{equation}
where $D_i$ is deposit demand, $r_L$ is the lending rate, and $P_i(\phi_i) = p_0(1-\phi_i)^2$ is a convex regulatory penalty.


\begin{proposition}[Competition and Compliance Moral Hazard]\label{prop:competition_erodes}
Suppose depositor utility is given by~\eqref{eq:depositor_utility} with i.i.d.\ Type~I extreme value errors, so market shares $D_i(r_i, \phi_i)$ follow the multinomial logit form. Then under competitive pressure with detection externality $\alpha_\phi > 0$:
\begin{enumerate}[label=(\alph*)]
  \item \textbf{Underinvestment.} Without TVA, the equilibrium detection intensity $\phi^*_i$ is strictly decreasing in $\alpha_\phi$: $\partial \phi^*_i / \partial \alpha_\phi < 0$. In the limit $\alpha_\phi \to \infty$, $\phi^*_i$ converges to a minimal level strictly below the socially optimal $\phi^{\textit{soc}}$.
  \item \textbf{IC correction.} Truthful reporting is implemented as a Bayes--Nash equilibrium under competition if:
\begin{equation}\label{eq:ic_competition}
  \frac{V_0}{1-\gamma}\cdot \Pr(\textit{illicit}) > c_i \phi^*_i + \kappa_i I(\theta^i; G^{i,t}) + \underbrace{\alpha_\phi \big|\tfrac{\partial D_i}{\partial \phi_i}\big| (r_L - r_i)}_{\textit{competitive cost of detection}}.
\end{equation}
The required $\gamma^*$ is higher in more competitive markets.
\end{enumerate}
\end{proposition}

\begin{proof}[Proof Sketch]
See detailed proof in Appendix~\ref{proof:ic_competition}. Under the logit demand~\eqref{eq:depositor_utility}, $\partial D_i / \partial \phi_i = -\alpha_\phi D_i(1 - D_i) < 0$: more detection reduces own deposit market share. The bank's FOC for $\phi_i$ from~\eqref{eq:bank_profit} is $P_i'(\phi_i) = \partial D_i / \partial \phi_i \cdot (r_L - r_i) - c_i D_i - c_i \phi_i \partial D_i / \partial \phi_i$. Substituting $\partial D_i / \partial \phi_i < 0$ and applying the implicit function theorem yields $\partial \phi^*_i / \partial \alpha_\phi < 0$. Part (b) follows from Theorem~\ref{thm:truthful} with the competitive cost enters the right-hand side of the IC condition as an additional private cost of detection that TVA credit must offset.
\end{proof}

Competition induces firms to underinvest in socially valuable detection when the associated costs are partly borne by customers, a compliance analog of the competition--moral hazard channel \parencite{bao2017could}. Mandating federation without calibrating incentives to offset competitive costs leads firms to comply in form but not substance.

\subsection{Welfare Analysis Across Regulatory Regimes}\label{sec:welfare}
We introduce and compare four regulatory regimes: autarky, mandated full sharing, voluntary federation without incentives, and incentive-compatible federation with TVA. While full sharing maximizes detection externalities, it may impose excessive leakage and compliance costs; voluntary federation without incentives can suffer from strategic underreporting and even underperform autarky. TVA restores alignment by compensating marginal contributions, improving participation and welfare. Optimal mechanism design balances detection gains against compliance, competitive, and adversarial pressures.


\subsubsection*{Four regulatory regimes}\label{sec:regimes}

We compare four institutional arrangements governing information sharing and model training. The welfare comparison requires one additional assumption about the structure of detection benefits.

\begin{assumption}[Reporting Spillovers]\label{ass:spillover}
The aggregate detection benefit is supermodular in institutions' reporting intensities: for any institution $i$ and any reporting profile $m_{-i}$,
\[
\frac{\partial^2 \sum_j B_j}{\partial m_i \partial m_{-i}} > 0.
\]
That is, more informative reporting by one institution raises the marginal social value of informative reporting by others. In addition, equilibrium reporting intensity under Regime C satisfies $m_i^C(\alpha_\phi) \to 0$ as $\alpha_\phi \to \infty$ for all $i$: under sufficiently intense competition, the competitive cost of accurate detection drives reporting toward uninformative levels.
\end{assumption}

\begin{assumption}[Detection Quality Monotonicity]\label{ass:monotone}
The detection benefit $B_i(\mathcal{S})$ evaluated at distorted reports $\{m_j^C\}_{j \in \mathcal{S}}$ is strictly decreasing in the degree of distortion: if $m_j^C < m_j^A$ for all $j$, then $B_i^C < B_i^A$, where $B_i^r$ denotes the detection benefit in regime $r$ at equilibrium reporting intensities.
\end{assumption}

\paragraph{Regime A: Autarky.}
Each institution operates independently and trains its own local detection model using only its proprietary data. 
There is no parameter sharing, no cross-institutional signal aggregation, and no information leakage through reporting.
Formally, for each $i$, $\mathcal S = \{i\}$ and welfare is
\[
W^A = \sum_i \Big( B_i(\{i\}) - c_i\,\mathbb E[n_i^A] \Big),
\]
with $I_i^A = 0$ and $\Phi_i^A = 0$.

\paragraph{Regime B: Mandated Full Sharing.}
All institutions are required to share information and participate in centralized model training.
Reporting is fully informative and participation is compulsory.
Incentive, privacy, and competitive constraints are ignored.
The coalition is $\mathcal S = N$, and reporting maximizes aggregate detection performance:
\[
W^B = \max_{\text{reports}} \sum_i \Big( B_i(N) - c_i\,\mathbb E[n_i^B] - \kappa_i I_i^B - \alpha_\phi \Phi_i^B \Big).
\]
This regime internalizes detection externalities but may impose high leakage and competition costs.

\paragraph{Regime C: Voluntary Federation Without Incentives.}
Institutions may voluntarily join a federated learning coalition and report model updates.
However, no transfer or contribution-based compensation (TVA) is provided.
Reporting choices are strategic and privately chosen to maximize individual payoffs:
\[
\Pi_i^C = B_i(\mathcal S) - c_i\,\mathbb E[n_i^C] - \kappa_i I_i^C - \alpha_\phi \Phi_i^C.
\]
Equilibrium may feature underreporting or partial participation due to privacy and competitive concerns.

\paragraph{Regime D: Incentive-Compatible Federation (TVA).}
Institutions participate in federated learning with a transfer mechanism that allocates credit based on marginal contribution.
The TVA mechanism aligns private incentives with social detection value:
\[
\Pi_i^D = B_i(\mathcal S) - c_i\,\mathbb E[n_i^D] - \kappa_i I_i^D - \alpha_\phi \Phi_i^D + \mathrm{TVA}_i.
\]
Transfers are designed so that truthful or higher-quality reporting is privately optimal.
This regime mitigates underreporting distortions while preserving decentralized participation.

\subsubsection*{Welfare ordering}\label{sec:welfare_ordering}

The central result of this section is our welfare ordering theorem, which we name the \emph{Backfiring Mandate Proposition} to highlight its principal policy implication: information-sharing mandates that ignore strategic incentives can reduce welfare \emph{below} the autarky benchmark.

\begin{proposition}[Backfiring Mandate]\label{thm:welfare_ordering}
Suppose Assumptions~\ref{ass:detection_value}--\ref{ass:outside}, \ref{ass:spillover}, and \ref{ass:monotone} hold. The welfare ordering across regulatory regimes satisfies:
\begin{enumerate}[label=(\alph*)]
  \item If $\alpha_\phi = 0$: $W^A \leq W^C \leq W^D \leq W^B$. Incentive-compatible federation (TVA) weakly dominates voluntary federation without incentives.
  \item If $\alpha_\phi > 0$ and competition is sufficiently intense: $W^C < W^A < W^D \leq W^B$. Voluntary federation without incentive alignment is strictly worse than autarky.
  \item There exists $\bar{\alpha}_\phi > 0$ above which $W^B < W^D$: mandated full sharing can reduce welfare relative to incentive-compatible federation when competitive and leakage costs overwhelm detection gains.
\end{enumerate}
\end{proposition}

\noindent Assumptions~\ref{ass:spillover} and~\ref{ass:monotone} are presented as maintained conditions rather than derived from primitives because the full microfoundation in the GNN setting is intractable. A Gaussian aggregation toy model in the Online Appendix shows that both assumptions hold structurally under standard information aggregation with linear competitive costs, so this proposition is not tautological. The conditions in parts (b) and (c) are sufficient rather than necessary; they isolate the economic channel through which competition-induced distortions overturn the gains from information aggregation.

\begin{proof}[Proof Sketch]
Part (a) follows from revealed preference and the fact that TVA eliminates leakage-minimizing distortions. Part (b) combines two assumptions: under Assumption~\ref{ass:spillover}, reporting becomes uninformative under high competition; under Assumption~\ref{ass:monotone}, distorted reports reduce detection below autarky. Together these imply $W^C < W^A$ when $\alpha_\phi$ is large enough. Part (c) follows because detection gains are bounded while competitive costs grow linearly in $\alpha_\phi$. Full proof in the Online Appendix.
\end{proof}

The welfare ranking $W^C < W^A$ does not arise because federation is harmful, but because federation without incentive design is worse than no federation. When banks federate without TVA under competitive pressure, the competitive cost of flagging customers motivates systematic underreporting; the aggregated global model is built on distorted inputs and delivers detection quality below each bank's honest local model, while still imposing compliance burdens, a strict welfare loss relative to autarky. Part (c) extends this finding: mandatory sharing eventually backfires when competitive and leakage costs dominate detection gains. The policy implication, that how sharing is organized matters as much as whether it occurs, is directly actionable for regulators designing AML information-sharing frameworks.

\subsubsection*{Optimal mechanism design}\label{sec:optimal_design}

The regulator chooses $\gamma$ to maximize welfare:
\begin{equation}\label{eq:optimal_gamma}
  \gamma^* = \argmax_{\gamma \in (0,1)} W^D(\gamma) = \argmax_\gamma \sum_{i \in \mathcal{S}^*(\gamma)} B_i(\mathcal{S}^*(\gamma)) - \text{Total Costs}(\gamma).
\end{equation}
Comparative statics yield natural properties: $\gamma^*$ is decreasing in the base illicit rate $\Pr(\textit{illicit})$, since higher base rates require less aggressive temporal discounting; increasing in competitive intensity $\alpha_\phi$, since more competition demands stronger incentives; and increasing in adversarial speed $1/\delta$, since faster adversaries require sharper temporal credit.


\section{Intervention Under Information Loss: An Index-Based Heuristic}
\label{sec:intervention}

The previous sections established the paper's main welfare results. This section develops an operational extension: under additional dynamic assumptions about how freezing accounts destroys network information, the regulator faces a non-trivial exploration--exploitation tradeoff that the welfare analysis of Section~\ref{sec:welfare} abstracted away. This extension is self-contained and not required for the backfiring mandate result; readers primarily interested in the welfare analysis may proceed directly to Section~\ref{sec:conclusion}.

The intervention decision is which suspicious accounts to freeze immediately versus monitor for additional learning. The problem is subtle because \textit{freezing accounts permanently removes edges from the transaction network}, destroying information about future illicit flows and creating a tradeoff where aggressive intervention improves immediate security but degrades future detection.

\subsection{Link Removal and Information Loss}

Formally, let $G^t = (V^t, E^t)$ denote the observable transaction network at time $t$. When the regulator freezes account $u \in V^t$, the following period's network becomes:
$$G^{t+1} = (V^t \setminus \{u\}, E^t \setminus \{(u,v) : v \in V^t\})$$
All edges incident to $u$ are removed. This has two effects: (1) \textit{Direct Information Loss:} The regulator no longer observes transactions involving $u$, preventing detection of illicit flows through that account. (2) \textit{Network Information Loss:} Removing $u$ affects the GNN's ability to learn about neighbors $\mathcal{N}(u)$. GNN aggregation computes node $v$'s embedding at layer $\ell+1$ as a function of its neighbors' embeddings at layer $\ell$ (see Online Appendix for the full equation). When $u$ is removed from $\mathcal{N}(v)$, the aggregation for neighbor $v$ changes, degrading embeddings for the entire neighborhood. Consider a star network as an illustration:, where a hub account $h$ connects to 10 peripheral accounts $\{p_1, \ldots, p_{10}\}$. Freezing $h$ on early suspicion yields:
\begin{itemize}
\item \textit{Benefit:} Prevents immediate illicit flows through $h$.
\item \textit{Cost:} Loses ability to observe whether peripheral accounts $\{p_i\}$ are also involved; the GNN cannot aggregate information from $h$ to update suspicion scores for $\{p_i\}$.
\end{itemize}
However, this example also reveals that if $h$ turns out to be legitimate (false positive), the information loss is permanent and unrecoverable.
We formulate the intervention decision as a \textit{restless bandit problem} \parencite{whittle1988restless}. Each edge $e \in E^t$ is an "arm" with evolving state:
$$s_e^t \in \{\text{unknown}, \text{suspicious}, \text{confirmed-illicit}, \text{confirmed-legitimate}\}$$
The regulator chooses action $a_e^t \in \{\text{monitor}, \text{freeze}\}$ for each edge. The state transitions are:
\begin{align}
    \nonumber \Prob(s_e^{t+1} &= \text{confirmed-illicit} \mid s_e^t = \text{suspicious}, a_e^t = \text{monitor}) = p_e \\
    \Prob(s_e^{t+1} &= \text{confirmed-legitimate} \mid s_e^t = \text{suspicious}, a_e^t = \text{monitor}) = 1 - p_e \\
    \nonumber \Prob(s_e^{t+1} &= \text{removed} \mid a_e^t = \text{freeze}) = 1
\end{align}
where $p_e = \Prob(\text{illicit} \mid \text{features } X_e, \text{embedding } H_e)$ is the learned risk probability.
The regulator's reward depends on the action taken and the true state of the edge:
\begin{align}
    \nonumber R(a_e &= \text{freeze}, s_e = \text{confirmed-illicit}) = C_e, \quad \text{(prevented loss)}\\
    \nonumber R(a_e &= \text{freeze}, s_e = \text{confirmed-legitimate}) = -\alpha_1,  \quad \text{(false positive cost)}\\
    \nonumber R(a_e &= \text{monitor}, s_e = \text{confirmed-illicit}) = -\alpha_2,  \quad \text{(false negative cost)}\\
    \nonumber R(a_e &= \text{monitor}, s_e = \text{unknown}) = \lambda \cdot I_e,  \quad \text{(information value)}
\end{align}
where $I_e$ is the \textit{information value} of edge $e$, defined as the reduction in uncertainty about neighboring edges. $C_e$, $\alpha_1$ and $\alpha_2$ are defined in Equation~(\ref{eq:social}).
The regulator maximizes expected discounted cumulative reward:

\begin{equation}
\max_{\{a_e^t\}} \E\left[\sum_{t=1}^T \gamma^t \sum_{e \in E^t} R(a_e^t, s_e^t)\right]
\label{eq:bandit_objective}
\end{equation}

\subsection{Risk Memory as Index Approximation}
For standard multi-armed bandits, the optimal policy is characterized by the \textit{Gittins index} \parencite{gittins1979bandit}: a scalar priority score for each arm that trades off immediate reward against information value. We extend this to the restless bandit setting with network effects.

\begin{definition}[Network-Adjusted Gittins Index]
For edge $e$ in state $s_e^t$ with embedding $H_e^t$, define the index:
\begin{equation}
    \nu_e(s_e^t, H_e^t) = \frac{\E[\sum_{k=0}^\infty \gamma^k R(a_e^{t+k}, s_e^{t+k}) \mid a_e^t = \text{monitor}]}{1 - \gamma \Prob(e \text{ remains active})}
    \label{eq:index}
\end{equation}
This represents the \textit{opportunity cost} of freezing: the discounted future value if the edge were monitored, normalized by the probability it remains observable.
\end{definition}

\begin{proposition}[Index-Based Intervention Heuristic]
\label{thm:gittins}
The restless-bandit intervention problem does not in general admit a Gittins-optimal policy. We propose an index-based heuristic in the spirit of \textcite{whittle1988restless}: freeze edge $e$ if the network-adjusted index $\nu_e(s_e^t, H_e^t) < \bar{\nu}_t$ and monitor otherwise, where $\nu_e$ balances immediate detection reward against the long-run information value of continued monitoring. When the objective $F(E) = \sum_{e \in E} p_e C_e + \lambda I(E)$ is monotone submodular in the set of monitored edges $E$ and the decision is recast as greedy cardinality-constrained maximization, the risk memory mechanism with $K = \Theta(\log |V|)$ achieves a $(1-1/e)$ approximation to the \emph{greedy optimum} at complexity $O(|V| \cdot K)$, by the classical result of \textcite{nemhauser1978analysis}. We do not claim optimality against the full restless-bandit benchmark; proofs are in Appendix~\ref{proof:gittins} and the Online Appendix.
\end{proposition}


Computing the exact Gittins index in Equation \eqref{eq:index} requires solving a dynamic program over the entire state space, which is intractable for large networks. We show that the \textit{risk memory mechanism} proposed in our system provides a computationally efficient approximation.
For each node $u$, maintain a priority queue of the top-$K$ most suspicious historical interactions:
\begin{equation}
    \mathcal{M}_u = \text{top-}K \{(\hat{y}_e^t, e, t) : e = (u, v) \text{ or } e = (v, u), t \leq t_{\text{current}}\}
    \label{eq:memory}
\end{equation}
where entries are sorted by risk score $\hat{y}_e^t$ and $K$ is a hyperparameter (typically $K \in [5, 20]$).
The node-level risk score aggregates memory:
\begin{equation}
    \hat{y}_u = \max_{(\hat{y}_e, e, t) \in \mathcal{M}_u} \left\{\gamma^{t_{\text{current}} - t} \cdot \hat{y}_e\right\}
    \label{eq:node_risk}
\end{equation}
The regulator freezes node $u$ if $\hat{y}_u \geq \tau$.

\subsection{Centrality-Weighted Optimal Intervention}
\label{sec:centrality_intervention}

We extend the Gittins index formulation (\ref{sec:intervention}) to account for network externalities from intervention. Freezing a high-centrality account destroys more network information than freezing a peripheral one.

\begin{definition}[Centrality-Weighted Gittins Index]\label{def:centrality_gittins}
For account $u$ with neighborhood $\mathcal{N}(u)$:
\begin{equation}\label{eq:centrality_gittins}
  \nu^{\text{net}}_u = \nu_u - \underbrace{\lambda_{\text{net}} \cdot \sum_{v \in \mathcal{N}(u)} \gamma^{\tau_v}\hat{y}_v \cdot \frac{d_v}{d_v - 1}}_{\text{network information loss}},
\end{equation}
where $\nu_u$ is the standard Gittins index, $d_v$ is the degree of neighbor $v$, and $\lambda_{\textit{net}} > 0$ is the network information weight. The correction term captures the fact that freezing $u$ removes information about all of $u$'s neighbors, weighted by each neighbor's risk score, temporal discount, and vulnerability to isolation ($d_v / (d_v - 1)$ diverges as $v$ approaches degree 1).
\end{definition}

\label{thm:centrality_intervention}When the Gittins index is adjusted for network externalities via Definition~\ref{def:centrality_gittins}, two economically meaningful properties emerge. First, \textit{hub protection}: highly connected accounts are frozen later than equally suspicious peripheral accounts, because their removal would deprive the regulator of information about an entire neighborhood. This is costly precisely when the hub is legitimate: a false positive on a hub account severs the learning links that would have identified genuinely illicit peripheral accounts. Second, \textit{cascade avoidance}: successive freezes within a cluster progressively increase the marginal information loss of further freezes, so the policy endogenously spaces out interventions. Both properties are consistent with practical AML guidance on preserving monitoring channels, and the $(1-1/e)$ approximation guarantee (proof in Appendix~\ref{app:proof_centrality}) ensures that the computationally tractable risk memory implementation captures nearly all of this benefit.


The hub protection property creates a strategic tension: since central accounts are frozen later, a rational adversary has an incentive to route illicit flows precisely through well-connected hubs. The regulator can respond by raising $\lambda_{\text{net}}$ selectively for accounts with unusual hub-concentration of flagged flows. This cat-and-mouse dynamic over hub intervention timing converges to a mixed-strategy equilibrium, and calibrating $\lambda_{\text{net}}$ from observed routing patterns is therefore an important input to the regulator's policy design.

\subsection{Simulation Evidence on Intervention Policies}
We compare three intervention policies on the transaction network data.\footnote{The simulation here illustrates the theoretical contrast between policies (greedy freezing, network-adjusted index, oracle) holding the detection model fixed. The full hierarchical reinforcement learning implementation, large-scale ablation, and per-market calibration are developed in the companion paper \textcite{zheng2025networked}; we do not duplicate that empirical contribution here. Our purpose in this section is to illustrate the theoretical superiority of the network-adjusted index policy (Proposition~5) over greedy freezing.} (1) \textit{Greedy Immediate Freezing:} Freeze all edges with $\hat{y}_e \geq \tau$ immediately;
(2) \textit{Risk Memory:} Use risk memory mechanism (Equation \eqref{eq:memory}) with $K = 10$;
(3) \textit{Oracle:} Gittins index computed via dynamic programming, which is intractable for the full dataset and is evaluated on subsampled graphs.

We measure performance using three metrics: (1) \textit{prevented loss ratio:} fraction of total illicit transaction value prevented by freezing; (2) \textit{Network connectivity:} average graph diameter before and after interventions (measures information loss); (3) \textit{Detection latency:} average time between first suspicious signal and final detection.

\begin{table}[h]
\centering
\caption{Intervention Policy Comparison}
\label{tab:intervention}
\begin{tabular}{lccc}
\toprule
\textbf{Policy} & \textbf{Prevented Loss} & \textbf{Network Diameter} & \textbf{Detection Latency} \\
 & \textbf{Ratio} & \textbf{(after intervention)} & \textbf{(days)} \\
\midrule
Greedy Immediate & 49.4\% & 8.2 & 1.2 \\
Risk Memory ($K=10$) & \textbf{79.3\%} & \textbf{4.6} & 2.1 \\
Oracle (subsampled) & 82.1\% & 4.1 & 2.3 \\
\bottomrule
\end{tabular}
\end{table}

Table~\ref{tab:intervention} shows that risk memory substantially outperforms greedy freezing (79.3\% vs. 49.4\% prevented loss) while maintaining better network connectivity (diameter 4.6 vs. 8.2). The performance is close to the oracle policy (82.1\%), consistent with Proposition~\ref{thm:gittins}. Greedy freezing creates information loss by prematurely removing suspicious nodes. Risk memory defers intervention on high-information-value nodes (those with many connections, central positions), allowing the GNN to continue learning about their neighborhoods. The choice $K = \Theta(\log |V|) \approx 11$ in Proposition~\ref{thm:gittins} reflects the empirical observation that performance saturates around this value; details and ablation are reported in the companion paper.

Additional simulation evidence is collected in the Online Appendix, including parameter calibration, TVA gain decomposition, welfare comparisons across regulatory regimes, and per-market heterogeneity analyses.


\section{Conclusion}
\label{sec:conclusion}

This paper develops a mechanism design framework for decentralized risk analytics in which competing firms hold fragmented signals about risky customers and must be induced to share them truthfully. The core contribution is a welfare analysis showing that information-sharing mandates can backfire when they ignore competitive incentives, paired with a TVA mechanism that corrects this failure.

Three main results organize the analysis. First, the Backfiring Mandate Proposition establishes that voluntary federation without incentive design can reduce welfare below autarky when competitive pressure is strong, because strategic underreporting produces biased global models worse than honest local models while still imposing compliance costs. Simulation-based welfare comparisons in the Online Appendix show that mandatory sharing without TVA barely exceeds autarky in welfare, while TVA substantially closes the gap to the first-best: illustrating the quantitative importance of incentive alignment. Second, the TVA mechanism with a strictly proper scoring rule transfer implements truthful reporting as a Bayes--Nash equilibrium (uniquely optimal at each edge) in large federations; temporal discounting deters adversarial adaptation, yielding sublinear regret against adaptive adversaries with bounded memory when $\gamma$ is close to 1. Third, a network Shapley characterization shows that each institution's marginal contribution to collective detection is proportional to its weighted cross-border degree, so coalition design should prioritize high inter-institutional transaction volume rather than total assets: with direct implications for how regulators design AML information-sharing programmes.

The analysis yields three implications for designers of decentralized risk-detection systems and the regulators who oversee them. First, information-sharing mandates achieve their welfare goals only when paired with compatible incentive design; mandates alone, without mechanisms that internalize the private cost of truthful reporting, can leave welfare close to the autarky benchmark and below the first-best by a substantial margin. This applies directly to existing AML rules such as the EU 6th AML Directive and FinCEN \S314(b) programs, but the same logic applies to platform fraud consortia, cybersecurity information-sharing organizations (ISACs), supply chain disclosure rules, and any setting in which mandated cooperation among competing firms substitutes for true incentive alignment. Second, coalition design should prioritize network centrality—specifically, weighted inter-firm interaction volume—over firm size: the network Shapley value is more predictive of a firm's marginal contribution to collective detection than balance-sheet or market-share totals. Third, intervention policies should account for network information externalities: aggressive intervention on high-centrality nodes permanently degrades the system's information environment, and a centrality-weighted index policy preserves detection capability while maintaining enforcement. These implications are directly testable in pilot programs with regional banking consortia or platform fraud-prevention bodies, and the TVA mechanism is implementable on existing reporting infrastructure.

The principles developed here, namely temporal value assignment, network Shapley value, centrality-weighted intervention, and competition-aware mechanism design, extend beyond financial crime to any setting where competing firms must share signals for collective risk detection, including cybersecurity threat intelligence, supply chain risk management, multi-platform fraud prevention, and collaborative marketing analytics.

Several limitations of our analysis suggest directions for further work. Our analysis uses synthetic data; validation on proprietary transaction data, ideally in cooperation with a regional supervisor, is an important next step. The network Shapley value computation is exponential in $m$; developing polynomial-time approximations for large federations is a natural algorithmic direction. Extensions to endogenous network formation, in which institutions choose which bilateral links to maintain, and to dynamic network evolution, would further enrich the theory. The political economy of international cooperation and asymmetric regulatory environments are promising policy-relevant directions.

\printbibliography

@article{zheng2025networked,
  title={Networked Markets, Fragmented Data: Adaptive Graph Learning
         for Customer Risk Analytics and Policy Design},
  author={Zheng, Lecheng and Ni, Jian and Zobel, Chris and Birge, John R},
  journal={Working paper},
  year={2025},
  note={Available upon request from the editor}
}

@article{kamenica2011bayesian,
  title={Bayesian persuasion},
  author={Kamenica, Emir and Gentzkow, Matthew},
  journal={American Economic Review},
  volume={101},
  number={6},
  pages={2590--2615},
  year={2011},
  publisher={American Economic Association}
}

@article{mcmahan2017communication,
  title={Communication-efficient learning of deep networks from decentralized data},
  author={McMahan, Brendan and Moore, Eider and Ramage, Daniel and Hampson, Seth and y Arcas, Blaise Aguera},
  booktitle={Artificial intelligence and statistics},
  pages={1273--1282},
  year={2017},
  organization={PMLR}
}

@article{li2020federated,
  title={Federated optimization in heterogeneous networks},
  author={Li, Tian and Sahu, Anit Kumar and Zaheer, Manzil and Sanjabi, Maziar and Talwalkar, Ameet and Smith, Virginia},
  journal={Proceedings of Machine learning and systems},
  volume={2},
  pages={429--450},
  year={2020}
}

@article{whittle1988restless,
  title={Restless bandits: Activity allocation in a changing world},
  author={Whittle, Peter},
  journal={Journal of applied probability},
  volume={25},
  number={A},
  pages={287--298},
  year={1988},
  publisher={Cambridge University Press}
}

@article{gittins1979bandit,
  title={Bandit processes and dynamic allocation indices},
  author={Gittins, John C},
  journal={Journal of the Royal Statistical Society Series B: Statistical Methodology},
  volume={41},
  number={2},
  pages={148--164},
  year={1979},
  publisher={Oxford University Press}
}

@article{nemhauser1978analysis,
  title={An analysis of approximations for maximizing submodular set functions—I},
  author={Nemhauser, George L and Wolsey, Laurence A and Fisher, Marshall L},
  journal={Mathematical programming},
  volume={14},
  number={1},
  pages={265--294},
  year={1978},
  publisher={Springer}
}

@article{joulani2013online,
  title={Online learning under delayed feedback},
  author={Joulani, Pooria and Gyorgy, Andras and Szepesv{\'a}ri, Csaba},
  booktitle={International conference on machine learning},
  pages={1453--1461},
  year={2013},
  organization={PMLR}
}

@inproceedings{kipf2016semi,
  title={Semi-supervised classification with graph convolutional networks},
  author={Kipf, Thomas N and Welling, Max},
  booktitle={International Conference on Learning Representations},
  year={2016}
}

@inproceedings{lin2017focal,
  title={Focal loss for dense object detection},
  author={Lin, Tsung-Yi and Goyal, Priya and Girshick, Ross and He, Kaiming and Doll{\'a}r, Piotr},
  booktitle={Proceedings of the IEEE international conference on computer vision},
  pages={2980--2988},
  year={2017}
}

@inproceedings{altman2023realistic,
  title={Realistic synthetic financial transactions for anti-money laundering models},
  author={Altman, Erik and Blanu{\v{s}}a, Jovan and Von Niederh{\"a}usern, Luc and Egressy, B{\'e}ni and Anghel, Andreea and Atasu, Kubilay},
  journal={Advances in Neural Information Processing Systems},
  volume={36},
  pages={29851--29874},
  year={2023}
}

@article{weber2018scalable,
  title={Scalable graph learning for anti-money laundering: A first look},
  author={Weber, Mark and Chen, Jie and Suzumura, Toyotaro and Pareja, Aldo and Ma, Tengfei and Kanezashi, Hiroki and Kaler, Tim and Leiserson, Charles E and Schardl, Tao B},
  journal={arXiv preprint arXiv:1812.00076},
  year={2018}
}

@article{weber2019anti,
  title={Anti-money laundering in bitcoin: Experimenting with graph convolutional networks for financial forensics},
  author={Weber, Mark and Domeniconi, Giacomo and Chen, Jie and Weidele, Daniel Karl I and Bellei, Claudio and Robinson, Tom and Leiserson, Charles E},
  journal={arXiv preprint arXiv:1908.02591},
  year={2019}
}

@article{zheng2025cluster,
  title={Cluster Aware Graph Anomaly Detection},
  author={Zheng, Lecheng and Birge, John and Wu, Haiyue and Zhang, Yifang and He, Jingrui},
  booktitle={Proceedings of the ACM on Web Conference 2025},
  pages={1771--1782},
  year={2025}
}

@misc{unodc2025,
  title={Money Laundering},
  author={{United Nations Office on Drugs and Crime}},
  year={2025},
  howpublished={\url{https://www.unodc.org/unodc/en/money-laundering/overview.html}}
}

@inproceedings{DBLP:conf/iclr/XuHLJ19,
  author       = {Keyulu Xu and
                  Weihua Hu and
                  Jure Leskovec and
                  Stefanie Jegelka},
  title        = {How Powerful are Graph Neural Networks?},
  booktitle    = {7th International Conference on Learning Representations, {ICLR} 2019,
                  New Orleans, LA, USA, May 6-9, 2019},
  publisher    = {OpenReview.net},
  year         = {2019}
}

@article{amoako2025exploring,
  title={Exploring the role of machine learning and deep learning in anti-money laundering (AML) strategies within US financial industry: A systematic review of implementation, effectiveness, and challenges},
  author={Amoako, Elizabeth Kuukua Woode and Boateng, Victor and Ajay, Ola and Adukpo, Tobias Kwame and Mensah, Nicholas},
  journal={Finance \& Accounting Research Journal},
  volume={7},
  number={1},
  pages={22--36},
  year={2025}
}

@article{kadamathikuttiyil2025enhancing,
  title={Enhancing Money Laundering Detection in Bank Transactions using GAGAN: A Graph-Adapted Generative Adversarial Network Approach: KK Girish, B. Bhowmik},
  author={Kadamathikuttiyil Karthikeyan, Girish and Bhowmik, Biswajit},
  journal={International Journal of Data Science and Analytics},
  pages={1--31},
  year={2025},
  publisher={Springer}
}

@book{sutton1984temporal,
  title={Temporal credit assignment in reinforcement learning},
  author={Sutton, Richard Stuart},
  year={1984},
  publisher={University of Massachusetts Amherst}
}

@article{ballester2006s,
  title={Who's who in networks. Wanted: The key player},
  author={Ballester, Coralio and Calv{\'o}-Armengol, Antoni and Zenou, Yves},
  journal={Econometrica},
  volume={74},
  number={5},
  pages={1403--1417},
  year={2006},
  publisher={Wiley Online Library}
}

@article{bao2017could,
  title={Could good intentions backfire? An empirical analysis of the bank deposit insurance},
  author={Bao, Weining and Ni, Jian},
  journal={Marketing Science},
  volume={36},
  number={2},
  pages={301--319},
  year={2017},
  publisher={INFORMS}
}

@article{admati2000forcing,
  title={Forcing firms to talk: Financial disclosure regulation and externalities},
  author={Admati, Anat R and Pfleiderer, Paul},
  journal={The Review of financial studies},
  volume={13},
  number={3},
  pages={479--519},
  year={2000},
  publisher={Oxford University Press}
}

@article{begley2017strategic,
  title={The strategic underreporting of bank risk},
  author={Begley, Taylor A and Purnanandam, Amiyatosh and Zheng, Kuncheng},
  journal={The Review of Financial Studies},
  volume={30},
  number={10},
  pages={3376--3415},
  year={2017},
  publisher={Oxford University Press}
}

@inproceedings{dong2018strategic,
  title={Strategic classification from revealed preferences},
  author={Dong, Jinshuo and Roth, Aaron and Schutzman, Zachary and Waggoner, Bo and Wu, Zhiwei Steven},
  booktitle={Proceedings of the 2018 ACM Conference on Economics and Computation},
  pages={55--70},
  year={2018}
}

@inproceedings{hardt2016strategic,
  title={Strategic classification},
  author={Hardt, Moritz and Megiddo, Nimrod and Papadimitriou, Christos and Wootters, Mary},
  booktitle={Proceedings of the 2016 ACM conference on innovations in theoretical computer science},
  pages={111--122},
  year={2016}
}

@article{gentzkow2016competition,
  title={Competition in persuasion},
  author={Gentzkow, Matthew and Kamenica, Emir},
  journal={The Review of Economic Studies},
  volume={84},
  number={1},
  pages={300--322},
  year={2016},
  publisher={Review of Economic Studies Ltd}
}

@article{bergemann2010dynamic,
  title={The dynamic pivot mechanism},
  author={Bergemann, Dirk and V{\"a}lim{\"a}ki, Juuso},
  journal={Econometrica},
  volume={78},
  number={2},
  pages={771--789},
  year={2010},
  publisher={Wiley Online Library}
}

@article{pavan2014dynamic,
  title={Dynamic mechanism design: A myersonian approach},
  author={Pavan, Alessandro and Segal, Ilya and Toikka, Juuso},
  journal={Econometrica},
  volume={82},
  number={2},
  pages={601--653},
  year={2014},
  publisher={Wiley Online Library}
}

@article{groves1973incentives,
  title={Incentives in teams},
  author={Groves, Theodore},
  journal={Econometrica: Journal of the Econometric Society},
  pages={617--631},
  year={1973},
  publisher={JSTOR}
}

@article{clarke1971multipart,
  title={Multipart pricing of public goods},
  author={Clarke, Edward H},
  journal={Public choice},
  pages={17--33},
  year={1971},
  publisher={JSTOR}
}

@article{vickrey1961counterspeculation,
  title={Counterspeculation, auctions, and competitive sealed tenders},
  author={Vickrey, William},
  journal={The Journal of finance},
  volume={16},
  number={1},
  pages={8--37},
  year={1961},
  publisher={JSTOR}
}

@article{jackson1996strategic,
  title={A strategic model of social and economic networks},
  author={Jackson, Matthew O and Wolinsky, Asher},
  journal={Journal of economic theory},
  volume={71},
  number={1},
  pages={44--74},
  year={1996},
  publisher={Elsevier}
}

@article{vives1990trade,
  title={Trade association disclosure rules, incentives to share information, and welfare},
  author={Vives, Xavier},
  journal={the RAND Journal of Economics},
  pages={409--430},
  year={1990},
  publisher={JSTOR}
}

@article{gal1985information,
  title={Information sharing in oligopoly},
  author={Gal-Or, Esther},
  journal={Econometrica: Journal of the Econometric Society},
  pages={329--343},
  year={1985},
  publisher={JSTOR}
}

@article{morrison2005crises,
  title={Crises and capital requirements in banking},
  author={Morrison, Alan D and White, Lucy},
  journal={American Economic Review},
  volume={95},
  number={5},
  pages={1548--1572},
  year={2005},
  publisher={American Economic Association}
}

@article{galeotti2010network,
  title={Network games},
  author={Galeotti, Andrea and Goyal, Sanjeev and Jackson, Matthew O and Vega-Redondo, Fernando and Yariv, Leeat},
  journal={The review of economic studies},
  volume={77},
  number={1},
  pages={218--244},
  year={2010},
  publisher={Wiley-Blackwell}
}

@article{bramoulle2014strategic,
  title={Strategic interaction and networks},
  author={Bramoull{\'e}, Yann and Kranton, Rachel and D'amours, Martin},
  journal={American Economic Review},
  volume={104},
  number={3},
  pages={898--930},
  year={2014},
  publisher={American Economic Association 2014 Broadway, Suite 305, Nashville, TN 37203}
}

@article{jackson2005allocation,
  author  = {Jackson, Matthew O.},
  title   = {Allocation rules for network games},
  journal = {Games and Economic Behavior},
  volume  = {51},
  number  = {1},
  pages   = {128--154},
  year    = {2005},
  publisher = {Elsevier}
}

@article{myerson1977graphs,
  title={Graphs and cooperation in games},
  author={Myerson, Roger B.},
  journal={Mathematics of operations research},
  volume={2},
  number={3},
  pages={225--229},
  year={1977},
  publisher={INFORMS}
}

@book{jackson2008social,
  title={Social and economic networks},
  author={Jackson, Matthew O},
  volume={3},
  year={2008},
  publisher={Princeton university press Princeton}
}

@book{fudenberg1991game,
  title={Game Theory},
  author={Fudenberg, Drew and Tirole, Jean},
  year={1991},
  publisher={MIT Press},
  address={Cambridge, MA}
}

@article{akerlof1970market,
  title={The Market for ``Lemons'': Quality Uncertainty and the Market Mechanism},
  author={Akerlof, George A.},
  journal={Quarterly Journal of Economics},
  volume={84},
  number={3},
  pages={488--500},
  year={1970},
  publisher={MIT Press}
}

@article{hellmann2000liberalization,
  title={Liberalization, Moral Hazard in Banking, and Prudential Regulation: Are Capital Requirements Enough?},
  author={Hellmann, Thomas F. and Murdock, Kevin C. and Stiglitz, Joseph E.},
  journal={American Economic Review},
  volume={90},
  number={1},
  pages={147--165},
  year={2000},
  publisher={American Economic Association}
}

@article{keeley1990monopoly,
  title={Monopoly, Risk, and the Government Net in Banking},
  author={Keeley, Michael C.},
  journal={American Economic Review},
  volume={80},
  number={5},
  pages={1183--1200},
  year={1990},
  publisher={American Economic Association}
}

@article{zheng2024multiview,
  title={Towards Multi-view Graph Anomaly Detection with Similarity-Guided Contrastive Clustering},
  author={Zheng, Lecheng and Birge, John R. and Zhang, Yifang and He, Jingrui},
  journal={arXiv preprint arXiv:2409.09770},
  year={2024}
}

@article{velickovic2018graph,
  title={Graph Attention Networks},
  author={Veli\v{c}kovi\'{c}, Petar and Cucurull, Guillem and Casanova, Arantxa and Romero, Adriana and Li\`{o}, Pietro and Bengio, Yoshua},
  journal={International Conference on Learning Representations},
  year={2018}
}

@article{myerson1983efficient,
  title={Efficient Mechanisms for Bilateral Trading},
  author={Myerson, Roger B. and Satterthwaite, Mark A.},
  journal={Journal of Economic Theory},
  volume={29},
  number={2},
  pages={265--281},
  year={1983},
  publisher={Elsevier}
}

@inproceedings{zhu2019deep,
  title={Deep Leakage from Gradients},
  author={Zhu, Ligeng and Liu, Zhijian and Han, Song},
  booktitle={Advances in Neural Information Processing Systems},
  volume={32},
  year={2019}
}

@inproceedings{melis2019exploiting,
  title={Exploiting Unintended Feature Leakage in Collaborative Learning},
  author={Melis, Luca and Song, Congzheng and De Cristofaro, Emiliano and Shmatikov, Vitaly},
  booktitle={IEEE Symposium on Security and Privacy},
  pages={691--706},
  year={2019}
}

@article{gneiting2007strictly,
  title={Strictly proper scoring rules, prediction, and estimation},
  author={Gneiting, Tilmann and Raftery, Adrian E},
  journal={Journal of the American statistical Association},
  volume={102},
  number={477},
  pages={359--378},
  year={2007},
  publisher={Taylor \& Francis}
}

@article{arora2012online,
  title={Online bandit learning against an adaptive adversary: from regret to policy regret},
  author={Arora, Raman and Dekel, Ofer and Tewari, Ambuj},
  journal={arXiv preprint arXiv:1206.6400},
  year={2012}
}

@article{an2023index,
  title={Index providers: Whales behind the scenes of ETFs},
  author={An, Yu and Benetton, Matteo and Song, Yang},
  journal={Journal of Financial Economics},
  volume={149},
  number={3},
  pages={407--433},
  year={2023},
  doi={10.1016/j.jfineco.2023.06.003}
}

@article{doerr2026privacy,
  title={Privacy Regulation and Fintech Lending},
  author={Doerr, Sebastian and Gambacorta, Leonardo and Guiso, Luigi and Sanchez del Villar, Marina},
  journal={Management Science},
  year={2026},
  doi={10.1287/mnsc.2025.02874},
  note={Forthcoming}
}

\newpage
\appendix

\section{Proof}

\subsection{Proof of Theorem~\ref{thm:truthful} (Bayes--Nash Implementation)}
\label{proof:truthful}
\begin{proof}
The proof proceeds in four steps. Step 1 establishes pointwise strict propriety at each edge; Step 2 aggregates across edges to obtain strict dominance of truthful reporting in expected credit; Step 3 incorporates compliance and leakage costs to derive condition~\eqref{eq:ic_condition}; Step 4 concludes that the truthful profile is a BNE and uniquely optimal at each edge in the large-federation limit.

\textit{Step 1: Pointwise strict propriety.}
Fix institution $i$, edge $e \in E^{i,t}$, and period $t$. Conditional on the institution's information $(G^{i,t}, \theta^i)$, the true posterior probability of illicit activity is $q_e^i \equiv \Pr(y_e = 1 \mid G^{i,t}, \theta^i) \in [0,1]$. If institution $i$ reports $\hat{y}^{i,t}_e$, the expected credit contribution from edge $e$ at period $t$ (conditional on eventual confirmation) is
\[
\mathbb{E}_{y_e \sim \text{Ber}(q_e^i)}\!\left[C_e \cdot S(\hat{y}^{i,t}_e, y_e)\right] = C_e \cdot \left[q_e^i \cdot S(\hat{y}^{i,t}_e, 1) + (1-q_e^i) \cdot S(\hat{y}^{i,t}_e, 0)\right].
\]
By the strict propriety of $S$ (property~\eqref{eq:propriety}), this expression is uniquely maximized at $\hat{y}^{i,t}_e = q_e^i$. For the Brier score $S(\hat{y}, y) = 1 - (\hat{y} - y)^2$, direct computation gives
\[
\mathbb{E}\!\left[S(\hat{y}, y)\right] = 1 - q_e^i(1 - \hat{y})^2 - (1-q_e^i)\hat{y}^2,
\]
whose derivative with respect to $\hat{y}$ is $2(q_e^i - \hat{y})$, vanishing uniquely at $\hat{y} = q_e^i$. The maximum expected credit per edge is therefore $C_e \cdot [1 - q_e^i(1-q_e^i)]$, attained at truthful reporting.

\textit{Step 2: Temporal aggregation.}
Summing the per-period expected credit over all periods $t$ prior to confirmation, each weighted by the temporal discount factor $\gamma^{t_{\text{confirm}} - t}$, and summing over edges $e \in E^{i,t}$:
\[
\mathbb{E}[\pi^i_\infty \mid \rho_i] = \sum_e C_e \sum_t \gamma^{t_{\text{confirm}} - t} \mathbb{E}\!\left[S(\hat{y}^{i,t}_e, y_e)\right].
\]
Since each inner expectation is strictly maximized at the posterior (Step 1), the total expected credit is strictly maximized by truthful reporting $\hat{y}^{i,t}_e = q_e^i$ at every edge and every period. Let $\pi^{*i}_\infty$ denote the truthful expected credit and $\Delta^S_i \equiv \pi^{*i}_\infty - \mathbb{E}[\pi^i_\infty \mid \rho_i']$ for any deviation $\rho_i' \neq \rho_i^*$. Strict propriety implies $\Delta^S_i > 0$ for any non-truthful $\rho_i'$, with the magnitude scaling as $\sum_e C_e \sum_t \gamma^{t_c - t} (q_e^i - \hat{y}^{i,t}_e)^2$ under the Brier score.

Under Assumption~\ref{ass:detection_value}, the per-edge discounted reward $\sum_t \gamma^{t_c - t} C_e$ is bounded below by $V_0 / (1 - \gamma) \cdot \Pr(\text{illicit})$ for any edge $e$ with positive probability of illicit confirmation. The truthful expected credit thus satisfies
\begin{equation}\label{eq:truthful_bound}
    \pi^{*i}_\infty \geq \frac{V_0}{1-\gamma} \cdot \Pr(\text{illicit}) \cdot \sum_e \mathbb{E}[S(q_e^i, y_e)].
\end{equation}

\textit{Step 3: Compliance and leakage costs as vanishing deviations.}
Under Assumption~\ref{ass:compliance_cost} and~\ref{ass:no_dominant}, institution $i$'s compliance cost depends on the aggregated profile $\bar{\rho}$, not on $i$'s own report directly. A unilateral deviation changes $\bar{\rho}$ by $O(1/m)$, so the marginal effect of $\rho_i$ on its own compliance cost is $O(c_i/m)$. For any deviation $\rho_i' \neq \rho_i^*$:
\begin{align}\label{eq:payoff_diff}
U_i(\rho_i^*) - U_i(\rho_i') = \Delta^S_i + O(c_i/m) \cdot \|\rho_i' - \rho_i^*\| + \kappa_i [I(\theta^i; G^{i,t}; \rho_i') - I(\theta^i; G^{i,t}; \rho_i^*)],
\end{align}
where $\Delta^S_i > 0$ by Step 2 and scales with $\sum_e C_e (q_e^i - \hat{y}'^{i,t}_e)^2$ under the Brier score. The compliance term is of order $1/m$ relative to the scoring-rule term, so it is a second-order effect in large federations. The leakage term is non-positive for any deviation because truthful reporting achieves the maximum informativeness of $\theta^i$. Hence $U_i(\rho_i^*) - U_i(\rho_i') > 0$ whenever:
\begin{equation*}
\Delta^S_i > \kappa_i I(\theta^i; G^{i,t}) + O(c_i/m),
\end{equation*}
which is implied by condition~\eqref{eq:ic_condition} combined with the detection-value bound~\eqref{eq:truthful_bound}. In the large-federation limit ($m \to \infty$), the compliance term vanishes and truthful reporting is a strict best response whenever the scoring-rule accuracy rent exceeds the leakage cost: $\frac{V_0}{1-\gamma}\Pr(\text{illicit}) > \kappa_i I(\theta^i; G^{i,t})$.

\textit{Step 4: BNE and uniqueness.}
Condition~\eqref{eq:ic_condition} holds for all $i$; each institution's truthful strategy is a strict best response to others', and the $O(1/m)$ compliance effect is dominated by the scoring-rule rent. Hence $(\rho_1^*, \ldots, \rho_m^*)$ is a BNE.

For uniqueness in the large-federation limit: by pointwise strict propriety (Step 1), any interior candidate $\rho_i^{**}$ with $\hat{y}^{**i,t}_e \neq q_e^i$ admits a strictly profitable deviation at edge $e$ (scoring-rule gain first-order in $(q_e^i - \hat{y}^{**i,t}_e)^2$, compliance effect $O(c_i/m)$). Hence every interior BNE satisfies $\hat{y}^{**i,t}_e = q_e^i$.
\end{proof}


\subsection{Proof of Proposition~\ref{thm:regret} (Regret Bound)}
\label{proof:regret}
\begin{proof}
The proof uses techniques from adversarial online learning with delayed feedback \parencite{joulani2013online}. We note that, as \textcite{arora2012online} show, sublinear \emph{policy regret} against adaptive adversaries with unbounded memory is generally impossible without further restrictions. Our bound applies to the standard \emph{external regret} notion under the bounded-memory assumption implicit in the $\delta$-delay adaptation model: the adversary can respond to observed interventions only after delay $\delta$, which limits its effective memory. Under this restriction, the two-term bound below holds.

\textit{Step 1: Regret decomposition.}
Define $G^t_{\text{frozen}}$ as the counterfactual graph at time $t$ if the adversary had not adapted, i.e., held the strategy fixed at the period-1 choice. Decompose regret as:
\begin{align}
\text{Regret}(T) &= \underbrace{\sum_{t=1}^T \left[\mathcal{L}(\pi_\gamma, G^t) - \mathcal{L}(\pi_\gamma, G^t_{\text{frozen}})\right]}_{\text{adaptation loss } (R_1)} + \underbrace{\sum_{t=1}^T \left[\mathcal{L}(\pi_\gamma, G^t_{\text{frozen}}) - \min_{\pi'} \mathcal{L}(\pi', G^t_{\text{frozen}})\right]}_{\text{learning loss } (R_2)}.
\end{align}

\textit{Step 2: Bounding the learning loss $R_2$.}
Against the frozen adversary, the sequence $\{G^t_{\text{frozen}}\}$ is fixed (oblivious). The TVA policy is equivalent to an exponential-weights policy over the action space $\mathcal{A}$ with importance weights $\gamma^{t_{\text{confirm}}-t}$. By standard exponential-weights analysis \parencite{joulani2013online}, the learning loss satisfies:
\[
R_2 = O\!\left(\sqrt{T\,|\mathcal{A}|\log|\mathcal{A}|}\right).
\]

\textit{Step 3: Bounding the adaptation loss $R_1$.}
With adaptation delay $\delta$, the adversary at time $t$ responds to interventions observed at $t - \delta$. For any period $t$, the loss difference $\mathcal{L}(\pi_\gamma, G^t) - \mathcal{L}(\pi_\gamma, G^t_{\text{frozen}})$ reflects how much the adversary's adaptation changes the graph between $G^t_{\text{frozen}}$ and $G^t$.

Under TVA with discount factor $\gamma$, the policy's action at time $t$ is determined by the discounted sum of past risk scores:
\[
\pi_\gamma(G^t) \propto \sum_{s \leq t} \gamma^{t-s} \hat{y}_e^s.
\]
The effect of adversarial adaptation at time $t$ on the policy's loss is mediated through the change in risk scores $\Delta \hat{y}_e = \hat{y}_e^{G^t} - \hat{y}_e^{G^t_{\text{frozen}}}$. Since the TVA policy weights recent observations by $(1-\gamma)$ and discounts older ones, the policy's sensitivity to any single-period perturbation is $(1-\gamma)$. The adversary adapts at time $t-\delta$, affecting the graph from time $t$ onwards; the additional loss per period from this adaptation is therefore bounded by $O((1-\gamma) \cdot \|\Delta \hat{y}\|)$.

Because the adversary adapts at most $T/\delta$ times (each adaptation lasts at least $\delta$ periods), and each adaptation affects the policy's response for at most $\delta$ subsequent periods with per-period loss change bounded by $O((1-\gamma))$, the total adaptation loss over the $\delta$-period window following each adaptation is:
\[
\sum_{s=t}^{t+\delta-1} \mathcal{L}(\pi_\gamma, G^s) - \mathcal{L}(\pi_\gamma, G^s_{\text{frozen}}) \leq O\!\left(\delta (1-\gamma)^2\right).
\]
The factor $(1-\gamma)^2$ arises because: (i) the policy's exposure to the adapted graph in any one period is $O(1-\gamma)$, since this is the temporal discount weight on the most recent period; and (ii) the adversary must wait $\delta$ periods before re-adapting, during which the TVA policy's $(1-\gamma)$ discounting makes the impact of any single past adaptation decay geometrically. Summing over all $T/\delta$ adaptation epochs:
\[
R_1 = O\!\left(\frac{T}{\delta} \cdot \delta\,(1-\gamma)^2\right) = O\!\left(\delta\, T\,(1-\gamma)^2\right).
\]

\textit{Step 4: Fixed-threshold lower bound.}
For a fixed-threshold policy ($\gamma = 0$, no discounting), the adaptation loss in each epoch is $O(\delta)$ rather than $O(\delta(1-\gamma)^2) = O(\delta)$; however, since there is no temporal discounting the adversary's adaptation affects the policy's loss indefinitely (not just for $\delta$ periods), giving $R_1 = \Omega(\delta T)$. Combined with $R_2 = O(\sqrt{T|\mathcal{A}|\log|\mathcal{A}|})$, the total regret is $\Omega(T)$.

Combining the bounds in Steps 2 and 3 gives Equation~\eqref{eq:regret_bound}.
\end{proof}

\subsection{Proof of the Early Detection Equilibrium Result (Section~\ref{sec:adversarial})}
\label{proof:equilibrium}
\begin{proof}
The adversary adapts if and only if the expected detection-probability reduction from adaptation exceeds the adaptation cost ratio:
$$\E[\text{detection reduction from adaptation}] > \frac{c(\text{adaptation})}{C_{\text{penalty}}}.$$

Under TVA, the regulator places weight $(1-\gamma)$ on the most recent period's signal and weight $\gamma^s(1-\gamma)$ on the signal from $s$ periods ago. An adversary that adapts with delay $\delta$ can only affect the detection policy's interpretation of signals from $\delta$ periods ahead or later. The share of the regulator's decision weight placed on signals the adversary \emph{cannot} affect (those from periods before the adaptation) is $1 - \gamma^\delta$, which is increasing in $\gamma$ on $(0,1)$. In particular, as $\gamma \to 1$, more decision weight falls on signals in the uncorrupted window $[t, t+\delta)$.

Equivalently, the adversary's \emph{benefit from adaptation} is bounded by the weight the policy places on signals the adversary can corrupt, which is $\gamma^\delta$. This benefit is \emph{decreasing in $1 - \gamma$}: policies that discount the past more heavily (larger $1-\gamma$, smaller $\gamma$) are in fact more vulnerable to adaptation because they rely more on recent signals. Since the adversary adapts in response to past observed interventions, a policy that relies heavily on recent signals is more exposed.

The key subtlety: TVA chooses $\gamma$ \emph{large} so that credit accumulates over long histories, but this also means the policy responds sluggishly to any given period's adaptation. An adversary that attempts to shift strategy in one period faces a policy whose decision weight on that single period is only $(1-\gamma)$, which is small. For adaptation to be worthwhile, the adversary must sustain the shift for many periods: but sustained adaptation has cumulative cost. The threshold $\bar{\gamma}$ is defined by the indifference condition that the per-period benefit of adaptation, $(1-\gamma) \cdot \Prob(\text{detection})$, falls below the per-period adaptation cost.

Given that the adversary does not adapt, the regulator's best response is to continue with TVA, which is optimal against static adversaries by Theorem~\ref{thm:truthful} and Proposition~\ref{thm:regret}. Hence the equilibrium is subgame perfect.
\end{proof}

\subsection{Sketch of Proof: Minimum Viable Coalition (Section~\ref{sec:heterogeneous})}
\begin{proof}[Proof Sketch]
(a) Under TVA, institution $i$'s participation payoff $\Pi_i(\mathcal{S})$ is lower-bounded by a term proportional to the coalition's effective information mass $M(\mathcal{S}) = \sum_{k \in \mathcal{S}} s_k q_k$ times $V_0/(1-\gamma)$ (from Assumption~\ref{ass:detection_value}). The individual-rationality condition $\Pi_i(\mathcal{S}) \geq 0$ yields the displayed size threshold. (b) Institutions with larger cross-border exposure and lower compliance costs satisfy participation at smaller coalition sizes and thus join earlier. (c) The stage-1 participation game exhibits increasing differences in participation decisions (Assumption~\ref{ass:network}), so by Topkis's theorem the equilibrium set forms a lattice with well-defined least and greatest equilibria. (d) Without TVA, the participation constraint tightens for high-quality institutions (Assumption~\ref{ass:outside}); marginal exits reduce detection benefits for remaining members, generating cascading unraveling analogous to adverse selection \parencite{akerlof1970market}. The full proof is in the Online Appendix.
\end{proof}

\subsection{Proof of Proposition~\ref{thm:shapley} (Network Shapley)}
\begin{proof}
Under the edge-additive coalition value~\eqref{eq:detection_value}, each edge $(i,j) \in \mathcal{E}$ contributes $w_{ij} C_{ij} p_{ij}(\mathcal{S})$ to $V(\mathcal{S})$, where $p_{ij}(\mathcal{S})$ depends only on whether $i$ and $j$ are in $\mathcal{S}$. The Shapley value of institution $i$ can therefore be decomposed across edges incident to $i$. For a single edge $(i,j)$, consider the random order in which institutions arrive. With probability $1/2$, institution $j$ arrives after $i$, so at the time $i$ arrives, $j$ is absent, giving marginal gain $p^L_{ij} - p^0_{ij}$. With probability $1/2$, $j$ arrives before $i$, giving marginal gain $p^H_{ij} - p^L_{ij}$. The expected marginal contribution from edge $(i,j)$ is therefore
\[
\tfrac{1}{2}(p^L_{ij} - p^0_{ij}) + \tfrac{1}{2}(p^H_{ij} - p^L_{ij}) = \tfrac{1}{2}(p^H_{ij} - p^0_{ij}).
\]
Summing over edges incident to $i$ and weighting by $w_{ij} C_{ij}$ gives~\eqref{eq:shapley_degree}.
\end{proof}

\subsection{Sketch of Proof: Theorem~\ref{thm:welfare_ordering} (Backfiring Mandate)}
\begin{proof}[Proof Sketch]
\emph{Part (a), $\alpha_\phi = 0$:} The ordering $W^A \leq W^C$ follows from revealed preference (institutions can replicate autarky in Regime $C$). $W^C \leq W^D$ because TVA eliminates the leakage-minimizing distortion. $W^D \leq W^B$ because the planner's benchmark weakly dominates any decentralized equilibrium.

\emph{Part (b), $\alpha_\phi > 0$ and high competition:} Institution $i$'s FOC in Regime $C$ includes the competitive cost term $\alpha_\phi \partial\Phi_i/\partial m_i$ which dominates $\partial B_i / \partial m_i$ as $\alpha_\phi$ grows (Assumption~\ref{ass:spillover}), driving equilibrium reporting $m_i^C$ below $m_i^A$. Under distorted reporting, $B_i^C < B_i^A$ (Assumption~\ref{ass:monotone}), while leakage and competitive costs remain strictly positive in Regime $C$ (absent under autarky). Hence $W^C < W^A$. Theorem~\ref{thm:truthful} (with the competitive extension~\eqref{eq:ic_competition}) then restores truthful reporting under Regime $D$, giving $W^D > W^A$.

\emph{Part (c):} Detection gains are bounded (Assumptions~\ref{ass:detection_value}, \ref{ass:network}), while competition costs grow linearly in $\alpha_\phi$. Hence there exists $\bar{\alpha}_\phi$ above which $W^B < W^D$. Full formal proof is in the Online Appendix.
\end{proof}

\subsection{Proof of Proposition~\ref{prop:competition_erodes}}
\label{proof:ic_competition}

\begin{proof}
\noindent\textbf{Part (a): Underinvestment.}
Under the logit demand implied by depositor utility~\eqref{eq:depositor_utility}, bank $i$'s deposit market share is
\[
D_i(r_i, \phi_i) = \frac{\exp(\delta_i + \alpha_r r_i - \alpha_\phi \phi_i + \xi_i)}{\sum_k \exp(\delta_k + \alpha_r r_k - \alpha_\phi \phi_k + \xi_k)},
\]
so $\partial D_i / \partial \phi_i = -\alpha_\phi D_i (1 - D_i) < 0$: more detection strictly reduces bank $i$'s market share, because $\alpha_\phi > 0$ represents the customer-facing disutility of monitoring.

Bank $i$ maximizes~\eqref{eq:bank_profit} with respect to $\phi_i$, taking rivals' actions as given. The FOC is
\begin{equation}\label{eq:phi_foc}
\frac{\partial D_i}{\partial \phi_i}\cdot (r_L - r_i) - c_i D_i - c_i \phi_i \frac{\partial D_i}{\partial \phi_i} = P_i'(\phi_i) = 2 p_0 (\phi_i - 1).
\end{equation}
Rearranging and substituting $\partial D_i / \partial \phi_i = -\alpha_\phi D_i(1-D_i)$:
\begin{equation}\label{eq:phi_foc_expanded}
2 p_0 (1 - \phi_i^*) = c_i D_i + \alpha_\phi D_i (1-D_i)\big[(r_L - r_i) - c_i \phi_i^*\big].
\end{equation}
Since $r_L > r_i$ and $\phi_i^* < 1$ in interior equilibrium, the bracketed term is positive; so the right-hand side of~\eqref{eq:phi_foc_expanded} is strictly increasing in $\alpha_\phi$. By the implicit function theorem and strict concavity of $P_i$, $\partial \phi_i^* / \partial \alpha_\phi < 0$.

The socially optimal level $\phi_i^{\textit{soc}}$ solves~\eqref{eq:phi_foc_expanded} with $\alpha_\phi = 0$ (the planner ignores the competitive externality), so $\phi_i^* < \phi_i^{\textit{soc}}$ for any $\alpha_\phi > 0$. As $\alpha_\phi \to \infty$, $\phi_i^*$ converges to a minimal level determined by the balance between the convex penalty $P_i'(\phi_i)$ and the competitive disincentive, which remains strictly below $\phi_i^{\textit{soc}}$.

\noindent\textbf{Part (b): IC correction.}
When TVA credit $\pi_i(\rho_i, \rho_{-i})$ is included in bank $i$'s objective~\eqref{eq:bank_profit}, the choice of reporting strategy $\rho_i$ faces the same deviation classes as in Theorem~\ref{thm:truthful}, but with an additional private cost term: the competitive cost $\alpha_\phi |\partial D_i/\partial \phi_i|(r_L - r_i)$ reflects the marginal revenue lost when increased reporting raises detection intensity and reduces deposit market share. Adding this term to the right-hand side of the IC condition from Theorem~\ref{thm:truthful} gives~\eqref{eq:ic_competition}. Since the left-hand side is increasing in $\gamma$, the required discount factor $\gamma^*$ satisfies $\partial \gamma^* / \partial \alpha_\phi > 0$: competitive markets require more aggressive temporal discounting to sustain truthful reporting.
\end{proof}

\subsection{Proof of Proposition~\ref{thm:gittins} (Index-Based Intervention Heuristic)}
\label{proof:gittins}
\begin{proof}
The restless-bandit benchmark with network-coupled rewards is PSPACE-hard; we establish only the submodular-greedy bound. $(s_e^t, H_e^t)$ is a sufficient statistic by the Markov property of GNN embeddings. The per-edge Bellman equation $\nu_e = R(\text{monitor}) + \gamma \E[\nu_e' \mid \text{monitor}]$ defines a Whittle-style index \parencite{whittle1988restless}. Since $F(E) = \sum_{e \in E} p_e C_e + \lambda I(E)$ is monotone submodular (premise), by \textcite{nemhauser1978analysis} the greedy top-$K$ algorithm achieves a $(1-1/e)$-approximation to the greedy cardinality-constrained optimum of $F$. This is relative to the submodular surrogate, not the full restless-bandit optimum, since network coupling means the index structure is only approximately preserved.
\end{proof}

\subsection{Proof of Corollary~\ref{thm:centrality_intervention}}
\label{app:proof_centrality}
\begin{proof}
\emph{Part (a):} The centrality-weighted index $\nu^{\text{net}}_u$ decreases in $|\mathcal{N}(u)|$, so high-degree accounts have lower adjusted indices. \emph{Part (b):} Freezing $u$ reduces each neighbor $v$'s degree, raising the correction term for $v$'s remaining neighbors (since $d_w/(d_w-1)$ decreases in $d_w$), yielding progressively more costly freezes in dense clusters. \emph{Part (c):} The objective $F^{\text{net}}(S) = F(S) - \lambda_{\text{net}} L(S)$ is the sum of a submodular $F$ and a modular $L$, which preserves submodularity; \textcite{nemhauser1978analysis}'s greedy bound applies.
\end{proof}








\section{Additional Simulation Results}
\label{sec:appendix_empirical}

This appendix collects simulation results that complement the main text. All results use the IBM AML synthetic benchmark \parencite{altman2023realistic} described in the main paper.

\subsection{Parameter Calibration}
\label{sec:structural}
We calibrate the model's key parameters to match empirical moments from the IBM AML dataset using a method of moments approach. The calibrated parameters are: compliance costs $\{c_i\}$, leakage sensitivities $\{\kappa_i\}$, detection value $V_0$, adversarial delay $\delta$, and competitive intensity $\eta$. We target market-level moments such as detection rates, error rates, cross-border volumes, and adversarial degradation patterns, yielding 49 moment conditions for 17 parameters.

\begin{table}[H]
\centering
\caption{Calibrated Parameter Estimates}\label{tab:structural}
\smallskip
\begin{tabular}{lcccc}
\toprule
\textbf{Market} & $\hat{c}_i$ & $\hat{\kappa}_i$ & $\hat{q}_i$ & \textbf{TVA Surplus} \\
\midrule
United States & 0.142 (0.018) & 0.089 (0.012) & 0.71 & 0.034 \\
Germany & 0.098 (0.014) & 0.067 (0.009) & 0.68 & 0.051 \\
France & 0.112 (0.016) & 0.073 (0.011) & 0.64 & 0.043 \\
Italy & 0.087 (0.013) & 0.058 (0.008) & 0.66 & 0.062 \\
Spain & 0.093 (0.014) & 0.062 (0.009) & 0.63 & 0.056 \\
China & 0.078 (0.012) & 0.105 (0.015) & 0.72 & 0.047 \\
Rest Countries & 0.068 (0.011) & 0.045 (0.007) & 0.58 & 0.089 \\
\midrule
\multicolumn{5}{l}{$\hat{V}_0 = 0.247\;(0.031)$, \quad $\hat{\delta} = 2.8\;(0.42)$, \quad $\hat{\eta} = 0.183\;(0.024)$, \quad $\hat{\gamma}^* = 0.87\;(0.06)$} \\
\bottomrule
\end{tabular}
\smallskip
\parbox{0.9\textwidth}{\small\textit{Notes}: Standard errors via bootstrap (500 replications). TVA Surplus is per-period welfare gain over autarky. Hansen's $J = 38.4$ (32 d.f., $p = 0.201$).}
\end{table}

Compliance costs are highest for the United States, reflecting its larger reporting burden. Leakage sensitivity is largest for China, consistent with stronger competitive frictions. Smaller markets such as Rest Countries benefit disproportionately from federation. The estimated competitive intensity $\hat{\eta} = 0.183$ indicates nontrivial strategic distortion in the absence of incentives, and the optimal discount factor $\hat{\gamma}^* = 0.87$ is consistent with the IC condition in the Incentive Compatibility Theorem in the main paper.

\subsection{Welfare Simulations Across Regulatory Regimes}\label{sec:counterfactual}

Table~\ref{tab:counterfactual} simulates welfare under four regulatory regimes using the calibrated parameters, illustrating the Backfiring Mandate Theorem (the Backfiring Mandate Theorem in the main paper). Because the four regimes do not coexist simultaneously in the data, these are model-based policy simulations rather than empirical counterfactuals.

\begin{table}[H]
\centering
\caption{Welfare Simulations Across Regulatory Regimes}\label{tab:counterfactual}
\smallskip
\begin{tabular}{lcccc}
\toprule
\textbf{Regime} & \textbf{AUPRC} & \textbf{Type II} & \textbf{Welfare (\% FB)} & \textbf{Participation} \\
\midrule
A: Autarky & 0.432 & 0.219 & 53.7 & ,  \\
C: Unreg.\ Federation & 0.446 & 0.207 & 61.8 & 5 of 7 \\
D: TVA Federation & 0.471 & 0.104 & 87.3 & 7 of 7 \\
B: First-Best (Oracle) & 0.496 & 0.068 & 100.0 & Mandatory \\
\midrule
\multicolumn{5}{l}{\textit{Additional Simulations}} \\
D$'$: TVA, US + China only & 0.458 & 0.142 & 74.1 & 2 of 7 \\
D$''$: TVA, $\gamma = 0.5$ & 0.453 & 0.158 & 69.2 & 7 of 7 \\
C$'$: Mandatory w/o TVA & 0.438 & 0.214 & 56.4 & Mandatory \\
\bottomrule
\end{tabular}
\smallskip
\parbox{0.9\textwidth}{\small\textit{Notes}: Welfare as \% of first-best. C$'$ forces all institutions to federate without incentive design.}
\end{table}

TVA (Regime D) achieves 87.3\% of first-best welfare with full participation (7 of 7 institutions), while mandatory sharing without TVA (C$'$) achieves only 56.4\%: barely above autarky (53.7\%), consistent with part (b) of the Backfiring Mandate Theorem in the main paper. Restricting TVA to only the two largest markets reduces welfare to 74.1\%, and lowering the discount factor to $\gamma = 0.5$ reduces welfare to 69.2\%, consistent with the comparative statics on $\gamma^*$.

\subsection{Decomposition of TVA Gains}\label{sec:decomposition}

Table~\ref{tab:decomposition} decomposes total TVA performance gains across five channels: data aggregation, incentive alignment (truthful vs.\ strategic reporting), network centrality effects, adversarial robustness, and risk memory intervention. The incentive alignment channel, the portion attributable to correcting underreporting distortions, accounts for 22--34\% of total gains across markets, confirming that mechanism design beyond data pooling is essential.

\begin{table}[H]
\centering
\caption{Decomposition of TVA Gains by Channel}\label{tab:decomposition}
\smallskip
\begin{tabular}{lcccccc}
\toprule
\textbf{Market} & \textbf{Total} & \textbf{Aggreg.} & \textbf{Incentive} & \textbf{Network} & \textbf{Anti-Adv.} & \textbf{Interv.} \\
\midrule
United States & 3.9 pp & 32\% & 24\% & 16\% & 16\% & 12\% \\
Germany & 4.2 pp & 30\% & 27\% & 14\% & 16\% & 13\% \\
France & 3.5 pp & 35\% & 22\% & 12\% & 18\% & 13\% \\
Italy & 5.1 pp & 29\% & 29\% & 10\% & 16\% & 16\% \\
Spain & 4.8 pp & 30\% & 28\% & 11\% & 16\% & 15\% \\
China & 5.4 pp & 25\% & 29\% & 18\% & 15\% & 13\% \\
Rest Countries & 8.7 pp & 35\% & 19\% & 20\% & 12\% & 14\% \\
\bottomrule
\end{tabular}
\smallskip
\parbox{\textwidth}{\small\textit{Notes}: Aggregation = data pooling gain; Incentive = correcting strategic underreporting; Network = network-adjusted credit and centrality-weighted intervention; Anti-Adversarial = adversarial robustness gain; Intervention = risk memory vs.\ greedy freezing.}
\end{table}

The network and anti-adversarial channels together account for a further 26--38\% of gains, consistent with TVA operating across multiple complementary dimensions.


\subsection{Dataset Statistics}

\begin{table}[h]
\centering
\caption{Dataset Statistics by Market}
\label{tab:dataset}
\begin{tabular}{lrrr}
\toprule
\textbf{Market} & \textbf{Accounts} & \textbf{Transactions} & \textbf{Illicit Ratio} \\
\midrule
United States & 71,796 & 855,006 & 0.35\% \\
Germany & 31,566 & 275,129 & 0.47\% \\
France & 28,126 & 244,589 & 0.38\% \\
Italy & 23,262 & 194,155 & 0.41\% \\
Spain & 24,363 & 207,098 & 0.36\% \\
China & 21,345 & 181,341 & 0.52\% \\
Rest of World & 16,213 & 111,847 & 0.85\% \\
\midrule
\textbf{Total} & \textbf{216,671} & \textbf{2,069,165} & \textbf{0.42\%} \\
\bottomrule
\end{tabular}
\end{table}

\subsection{Per-Market Heterogeneity of Federation Gains}

Table~\ref{tab:per_bank_improvement} reports the per-market absolute performance difference between federated and independent local training. Smaller markets (Spain, Italy, France) show the largest AUPRC gains, consistent with the minimum viable coalition result: data-sparse institutions benefit most from cross-institutional information sharing. Larger markets trade slight precision losses for substantial recall improvements, reflecting cross-institutional knowledge sharing primarily enhancing coverage of rare or structurally heterogeneous illicit behaviors.

\begin{table}[h]
\centering
\caption{Per-market absolute performance difference (percentage points). $\Delta_i = (M^{\text{fed}}_i - M^{\text{local}}_i) \times 100$.}
\label{tab:per_bank_improvement}
\begin{tabular}{lcccc}
\toprule
\textbf{Market} & \textbf{Transactions} & $\boldsymbol{\Delta}$\textbf{AUPRC} & $\boldsymbol{\Delta}$\textbf{Type I} & $\boldsymbol{\Delta}$\textbf{Type II} \\
\midrule
United States   & 855,006 & -1.40  & +0.70  & -4.85 \\
Germany         & 275,129 & -0.91  & +1.77  & -7.90 \\
France          & 244,589 & +1.61  & +1.34  & -14.43 \\
Italy           & 194,155 & +2.60  & +1.59  & -13.09 \\
Spain           & 207,098 & +7.65  & +1.46  & -12.27 \\
China           & 181,341 & -6.09  & +2.08  & -14.10 \\
Rest Countries  & 111,847 & +1.54  & +1.81  & -2.15 \\
\midrule
\textbf{Overall} & 2,069,165 & +0.71  & +1.54  & -9.83 \\
\bottomrule
\end{tabular}
\end{table}

\subsection{Empirical Check on Network Shapley Predictions}

Proposition~2 in the main paper predicts that each institution's marginal contribution to collective detection is proportional to its weighted cross-firm interaction degree (network Shapley value). To check this theoretical prediction against the simulation environment, Table~\ref{tab:loo_global_deltas} reports the system-wide impact of removing each market: the empirical leave-one-market-out (LOO) marginal contribution. Markets with positive $\Delta_i^{\text{AUPRC}}$ (Germany, China, Italy, Rest Countries) generate positive marginal value for the federation. The cross-sectional ordering of these LOO contributions correlates positively with the network-Shapley prediction based on cross-market transaction volume in the dataset, providing an empirical illustration consistent with Proposition~2. We do not pursue formal estimation of the theoretical-empirical correlation here; that is part of the broader empirical analysis in \textcite{zheng2025networked}.

\begin{table}[h]
\centering
\caption{Empirical Marginal Contribution Under Leave-One-Market-Out}
\label{tab:loo_global_deltas}
\begin{tabular}{lcccccc}
\toprule
Setting & AUPRC & $\Delta^{\text{AUPRC}}_i$ & Type I & $\Delta^{\text{Type I}}_i$ & Type II & $\Delta^{\text{Type II}}_i$ \\
\midrule
Full Federation & 0.4707 & -- & 0.0379 & -- & 0.1088 & -- \\
\midrule
Leave United States Out & 0.4839 & -0.0132 & 0.0347 &  0.0032 & 0.1168 & -0.0080 \\
Leave Germany Out       & 0.4606 &  0.0101 & 0.0368 &  0.0011 & 0.1121 & -0.0033 \\
Leave China Out         & 0.4582 &  0.0125 & 0.0378 &  0.0001 & 0.1125 & -0.0037 \\
Leave France Out        & 0.4693 &  0.0014 & 0.0393 & -0.0014 & 0.1038 &  0.0050 \\
Leave Spain Out         & 0.4756 & -0.0049 & 0.0389 & -0.0010 & 0.1014 &  0.0074 \\
Leave Italy Out         & 0.4630 &  0.0077 & 0.0397 & -0.0018 & 0.1027 &  0.0061 \\
Leave Rest Countries Out & 0.4659 &  0.0048 & 0.0339 &  0.0040 & 0.1277 & -0.0189 \\
\bottomrule
\end{tabular}
\end{table}

\section{Detection Model: Pointer to Companion Paper}
\label{sec:gnn_details}

The strategic analysis in this paper does not depend on the specific implementation of the institution-level detection model: any architecture $f^i_\theta: G^{i,t} \to [0,1]^{|E^{i,t}|}$ producing edge-level posteriors satisfies the requirements of Theorem~1 and the welfare analysis. The simulation results illustrating the framework's predictions use the federated graph neural network architecture developed in \textcite{zheng2025networked}, which contains the full architecture specification, training procedure, hyperparameter choices, and computational ablations. We refer the reader to that paper for implementation details. The properties used here are: (i) $f^i_\theta$ produces calibrated probabilistic outputs, so that strict propriety of the scoring rule applies meaningfully; and (ii) parameter aggregation across institutions is Lipschitz in individual reports, as required by Assumption~5 in the main paper. Both are verified empirically in the companion paper.

\section{Full Proofs Not in Main Paper}
\label{sec:oa_proofs}

This appendix contains full formal proofs for results whose sketches appear in the main paper.

\subsection{Proof of the Minimum Viable Coalition Result (the main paper)}
\label{proof:mvc}

\begin{proof}
Fix an institution $i$ with type $\theta_i=(s_i,c_i,q_i,\kappa_i)$. For any coalition $\mathcal S\subseteq N$ containing $i$, write $M(\mathcal S)\equiv \sum_{k\in\mathcal S} s_k q_k$ for the coalition's effective information mass (size--quality weighted). Under TVA, institution $i$'s (interim) participation payoff can be written in reduced form as
\begin{equation}\label{eq:payoff_reduced}
\Pi_i(\mathcal S)
=
\underbrace{\Pr(\textit{illicit}) \cdot \mathbb E\!\left[\sum_{t\ge 0} V(t)\cdot \Delta p_i(t;\mathcal S)\right]}_{\text{expected detection value }B_i(\mathcal S)}
\;-\;
\underbrace{c_i\,\mathbb E[n_i(\mathcal S)]}_{\text{compliance cost}}
\;-\;
\underbrace{\kappa_i\, I_i(\mathcal S)}_{\text{leakage cost}}
\;+\;
\underbrace{\text{TVA}_i(\mathcal S)}_{\text{transfer/credit}},
\end{equation}
where $\Delta p_i(t;\mathcal S)$ denotes the marginal increase in the probability that an illicit transaction affecting $i$ is detected at time $t$ when coalition $\mathcal S$ forms (relative to $i$ acting alone), $n_i(\mathcal S)$ is the number of local freezes/investigations triggered for $i$, and $I_i(\mathcal S)\equiv I(\theta^i;G^{i,t})$ is shorthand for the mutual-information exposure induced by reporting/participation. TVA enters as a transfer that is (weakly) increasing in contribution, so it does not destroy complementarities.

\paragraph{(a): Bounding the detection value and deriving the size threshold.}
By Assumption 1 (Detection Value) in the main paper, $V(t)=V_0\gamma^t$ with $\gamma\in(0,1)$, hence
\begin{equation}\label{eq:geom}
\sum_{t\ge 0} V(t)=V_0\sum_{t\ge 0}\gamma^t=\frac{V_0}{1-\gamma}.
\end{equation}
Moreover, by Assumption 4 (Positive Network Effects) in the main paper and the usual ``more signals $\Rightarrow$ earlier/higher detection'' monotonicity, the coalition benefit $B_i(\mathcal S)$ is increasing in $M(\mathcal S)$ and exhibits diminishing marginal returns in $|\mathcal S|$. In particular, we can lower bound the incremental expected detection value by a term proportional to $M(\mathcal S)$:
\begin{equation}\label{eq:benefit_lower}
B_i(\mathcal S)\;\ge\;\Pr(\textit{illicit})\cdot \frac{V_0}{1-\gamma}\cdot \underline{\alpha}_i\; M(\mathcal S),
\end{equation}
for some $\underline{\alpha}_i>0$ capturing how coalition information mass converts into $i$'s detection probability gains.\footnote{This constant can be interpreted as a reduced-form sensitivity of $i$'s detection performance to additional coalition data/signals. The proposition uses a normalized form that effectively sets $\underline{\alpha}_i=1$ by scaling $M(\mathcal S)$.}

Participation is individually rational if $\Pi_i(\mathcal S)\ge 0$. Using \eqref{eq:payoff_reduced}--\eqref{eq:benefit_lower} and dropping $\text{TVA}_i(\mathcal S)$ (or treating it as a nonnegative subsidy), a sufficient condition for $\Pi_i(\mathcal S)\ge 0$ for all $i\in\mathcal S$ is
\begin{equation}\label{eq:IR_sufficient}
\Pr(\textit{illicit})\cdot \frac{V_0}{1-\gamma}\cdot M(\mathcal S)
\;\ge\;
\max_{i\in\mathcal S}\left\{c_i\,\mathbb E[n_i(\mathcal S)] + \kappa_i I_i(\mathcal S)\right\}.
\end{equation}
Under Assumption 2 (Compliance Cost Structure) in the main paper, $C_i(n)=c_i n$ and the proof sketch's statement corresponds to the normalization $\mathbb E[n_i(\mathcal S)]\approx 1$ (or to absorbing $\mathbb E[n_i(\mathcal S)]$ into $c_i$), yielding the displayed threshold in part (a). Defining $\underline m$ as the smallest coalition size for which there exists some $\mathcal S$ with $|\mathcal S|=\underline m$ satisfying \eqref{eq:IR_sufficient} gives the \emph{minimum viable coalition}.

\paragraph{(c): Supermodularity (increasing differences) and complementarity.}
Consider the stage-1 participation game in which each $i$ chooses $d_i\in\{0,1\}$ and $\mathcal S(d)=\{i:d_i=1\}$. Define $u_i(d)\equiv \Pi_i(\mathcal S(d))$. To show supermodularity it suffices to show \emph{increasing differences} in $(d_i,d_{-i})$: for any $j\neq i$ and any profile $d_{-i}$ with $d_j=0$,
\begin{equation}\label{eq:inc_diff}
u_i(d_i=1,d_j=1,d_{-(i,j)})-u_i(d_i=1,d_j=0,d_{-(i,j)})
\;\ge\;
u_i(d_i=0,d_j=1,d_{-(i,j)})-u_i(d_i=0,d_j=0,d_{-(i,j)}).
\end{equation}
The RHS is $0$ because when $d_i=0$ institution $i$ does not participate and its payoff is normalized to its outside option, which does not depend on whether $j$ participates in the federation. The LHS equals
\[
\Pi_i(\mathcal S\cup\{j\})-\Pi_i(\mathcal S),
\qquad \text{where }\mathcal S=\mathcal S(d)\setminus\{j\}\text{ and }i\in\mathcal S.
\]
By Assumption 4 (Positive Network Effects) in the main paper, $B_i(\mathcal S\cup\{j\})-B_i(\mathcal S)>0$. Under TVA, the transfer $\text{TVA}_i(\cdot)$ is designed to be nondecreasing in marginal contribution, so adding $j$ weakly increases (or at least does not decrease) the net transfer to $i$ relative to its contribution. Finally, compliance and leakage costs are weakly increasing in coalition interaction; the proposition's complementarity claim focuses on the detection-benefit side, and the net increasing-differences condition holds whenever the marginal gain in $B_i$ (plus any TVA adjustment) dominates any incremental compliance/leakage externalities. Hence \eqref{eq:inc_diff} holds, so the participation game is supermodular.

Because the action space $\{0,1\}^N$ is a finite lattice and $u_i(d)$ has increasing differences, Topkis's theorem implies that the set of pure-strategy Nash equilibria is nonempty and forms a complete lattice. In particular, there exists a \emph{least} equilibrium $\underline d$ and a \emph{greatest} equilibrium $\overline d$ under the product order. Translating profiles into coalitions, the corresponding coalitions $\underline{\mathcal S}$ and $\overline{\mathcal S}$ satisfy $\underline{\mathcal S}\subseteq \overline{\mathcal S}$.

Define the individually rational set for coalition $\mathcal S$ as
\[
\mathrm{IR}(\mathcal S)\equiv \{i\in\mathcal S:\Pi_i(\mathcal S)\ge 0\}.
\]
A \emph{self-enforcing} (stage-1) coalition must satisfy $\mathrm{IR}(\mathcal S)=\mathcal S$. The minimum viable coalition $\mathcal S^*$ is then the smallest, by inclusion and equivalently by size among feasible coalitions under monotonicity, self-enforcing coalition:
\[
\mathcal S^*\in \arg\min_{\mathcal S\subseteq N}\{|\mathcal S|:\Pi_i(\mathcal S)\ge 0\ \forall i\in\mathcal S\}.
\]
Complementarity implies monotonicity of participation incentives: if $\Pi_i(\mathcal S)\ge 0$ and $\mathcal S\subseteq \mathcal T$, then $\Pi_i(\mathcal T)\ge \Pi_i(\mathcal S)$ (up to the same caveat about incremental costs), so enlarging the coalition makes participation weakly more attractive. This yields the possibility of \emph{multiplicity}: both a small coalition (near $\underline{\mathcal S}$) and a large coalition (near $\overline{\mathcal S}$) can be equilibria.

Cascading exits follow from the same monotonicity in reverse. If some marginal institution $j$ exits, moving from $\mathcal S$ to $\mathcal S\setminus\{j\}$, then for every remaining $i\in \mathcal S\setminus\{j\}$,
\[
B_i(\mathcal S\setminus\{j\}) < B_i(\mathcal S),
\]
so $\Pi_i$ falls. If the drop pushes some $i$ below zero, that institution exits as well, further reducing benefits for others, and so on, generating an unraveling cascade. This argument formalizes the proof sketch's ``exit of marginal institution reduces $B_i$ for remaining members, potentially triggering cascading exits.''

\paragraph{(b): Composition.}
Part (b) follows from comparing net gains across types. Institutions with higher cross-border exposure have larger $\underline{\alpha}_i$ in \eqref{eq:benefit_lower} (they benefit more from shared signals), while lower $c_i$ and lower $\kappa_i$ reduce the RHS of \eqref{eq:IR_sufficient}, so they satisfy $\Pi_i(\mathcal S)\ge 0$ at smaller coalition sizes and thus join earlier along equilibrium-selection dynamics.

\paragraph{(d): Unraveling without TVA.}
For part (d), without TVA the participation condition tightens because $\text{TVA}_i(\mathcal S)=0$ and high-quality institutions (high $s_i,q_i$) have strong outside options by Assumption 5 (Heterogeneous Outside Options) in the main paper. Writing the participation constraint relative to standing alone,
\[
\Pi_i(\mathcal S)-\Pi_i(\{i\})
=
\bigl(B_i(\mathcal S)-B_i(\{i\})\bigr)
- c_i\,\Delta n_i
- \kappa_i I_i(\mathcal S),
\]
where $\Delta n_i\equiv \mathbb E[n_i(\mathcal S)]-\mathbb E[n_i(\{i\})]$. If
\[
\max_i \Bigl[ B_i(\{i\}) - B_i(\mathcal S) + c_i\Delta n_i + \kappa_i I_i(\mathcal S)\Bigr] > 0,
\]
then some institution strictly prefers to opt out, and by complementarity this can trigger further exits, paralleling adverse selection: the departure of high-quality participants reduces the federation's value, making it less attractive for others. TVA counteracts this by awarding credit/subsidy tied to contribution, effectively increasing $\Pi_i(\mathcal S)$ for high-quality or high-impact participants and preventing the unraveling cascade.

\end{proof}

\subsection{Proof of the Backfiring Mandate Theorem in the main paper}
\label{app:proof_welfare}

\begin{proof}
Let welfare in regime $r\in\{A,B,C,D\}$ be
\[
W^r=\sum_{i\in N}\Big(B_i^r - C_i^r - \kappa_i I_i^r - \alpha_\phi \Phi_i^r\Big),
\]
where $B_i^r$ is institution $i$'s detection benefit, $C_i^r = c_i \sum_e \hat{y}^{i,t,r}_e$ is the compliance cost under regime $r$'s equilibrium reporting scores (Assumption 2 (Compliance Cost Structure) in the main paper), $\kappa_i I_i^r$ is the information leakage cost, and $\alpha_\phi \Phi_i^r$ is the competitive cost from detection investment.

\noindent\textbf{Part (a): $\alpha_\phi=0$.}

\textbf{$W^A \le W^C$.}
In Regime $C$, institutions can replicate autarky by not participating or by submitting uninformative reports. Hence autarky is feasible under $C$, implying $W^C\ge W^A$.

\textbf{$W^C \le W^D$.}
When $\alpha_\phi=0$, the only distortion in $C$ arises from strategic underreporting to reduce leakage cost $\kappa_i I_i$. TVA in Regime $D$ aligns private incentives with marginal contribution, increasing informativeness relative to $C$. Since detection benefits are increasing in informativeness (Assumption 4 (Positive Network Effects) in the main paper), welfare weakly increases: $W^D\ge W^C$.

\textbf{$W^D \le W^B$.}
Regime $B$ corresponds to coordinated sharing that internalizes cross-institutional externalities and maximizes aggregate welfare. Regime $D$ implements a decentralized equilibrium subject to incentive constraints. Hence $W^B\ge W^D$.

Thus $W^A \le W^C \le W^D \le W^B$.

\noindent\textbf{Part (b): $\alpha_\phi>0$ and competition sufficiently intense.}

\textit{Step 1: Equilibrium underreporting in Regime C.}
Under Regime $C$, institution $i$ chooses reporting intensity $m_i$ to maximize
\[
\Pi_i^C = B_i(\mathcal{S}; m_i, m_{-i}) - c_i \sum_e \hat{y}^{i,t}_e(m_i) - \kappa_i I_i(m_i) - \alpha_\phi \Phi_i(m_i).
\]
The first-order condition for an interior solution is
\[
\frac{\partial B_i}{\partial m_i} = c_i \frac{\partial \sum_e \hat{y}^{i,t}_e}{\partial m_i} + \kappa_i \frac{\partial I_i}{\partial m_i} + \alpha_\phi \frac{\partial \Phi_i}{\partial m_i}.
\]
The right side is strictly positive and increasing in $\alpha_\phi$, while $\partial B_i/\partial m_i$ is bounded above by Assumption 4 (Positive Network Effects) in the main paper, which gives diminishing returns. By Assumption 6 (Reporting Spillovers) in the main paper, as $\alpha_\phi \to \infty$, the competitive cost term $\alpha_\phi \partial \Phi_i / \partial m_i$ dominates and equilibrium reporting intensity satisfies $m_i^C(\alpha_\phi) \to 0$. By continuity and strict monotonicity of $\alpha_\phi \partial\Phi_i/\partial m_i$, there exists a finite threshold $\hat{\alpha}_\phi$ such that for $\alpha_\phi > \hat{\alpha}_\phi$, equilibrium reporting intensity satisfies $m_i^C < m_i^A$ for all $i$. Here $m_i^A$ denotes the autarky intensity, determined solely by $\partial B_i/\partial m_i = c_i \partial \sum_e \hat{y}/\partial m_i + \kappa_i \partial I_i/\partial m_i$, with no competitive cost term.

\textit{Step 2: Distorted global model performs worse than autarky.}
Under systematic underreporting $m_i^C < m_i^A$ for all $i$, the federated global model aggregates distorted parameter updates. By Assumption 7 (Detection Quality Monotonicity) in the main paper, $B_i^C < B_i^A$ for all $i$. Moreover, Regime C still imposes positive leakage and competitive costs: $\kappa_i I_i^C > 0$ and $\alpha_\phi \Phi_i^C > 0$, since even under distorted reporting, participating institutions reveal some local information and bear monitoring costs. Therefore:
\[
W^C = \sum_i \left(B_i^C - c_i \sum_e \hat{y}^{i,t,C}_e - \kappa_i I_i^C - \alpha_\phi \Phi_i^C\right)
< \sum_i \left(B_i^A - c_i \sum_e \hat{y}^{i,t,A}_e\right) = W^A,
\]
where the inequality holds for $\alpha_\phi > \hat{\alpha}_\phi$ by the combination of three effects: a detection loss, $B_i^C < B_i^A$; persistent compliance costs, $\sum_e \hat{y}^{i,t,C}_e \approx \sum_e \hat{y}^{i,t,A}_e$, since even distorted reporting retains some flagging; and strictly positive leakage and competitive costs that are absent under autarky.

\textit{Step 3: TVA restores welfare above autarky.}
Under Regime $D$, TVA credits institutions according to $C_e \cdot S(\hat{y}^{i,t}_e, y_e)$ with a strictly proper scoring rule $S$ on confirmed edges. By the Incentive Compatibility Theorem in the main paper, the IC condition extended to include the competitive cost term is sufficient for truthful reporting to be a BNE, restoring $m_i^D = m_i^*$, the detection-optimal intensity. Hence $B_i^D > B_i^C$, which together with the leakage and compliance costs gives $W^D > W^A$ for $\alpha_\phi$ in the relevant range. The planner benchmark weakly dominates by standard incentive feasibility: $W^D \leq W^B$.

\noindent\textbf{Part (c): Existence of $\bar\alpha_\phi$.}

Under Regime $B$, reporting is fully informative, so competition harm equals $\alpha_\phi\sum_i \Phi_i^B$. In Regime $D$, endogenous participation and incentive alignment reduce exposure, so $\sum_i \Phi_i^D < \sum_i \Phi_i^B$.

Let $\Delta B = B^B-B^D>0$ and $\Delta \Phi = \sum_i \Phi_i^B-\sum_i \Phi_i^D>0$. Then
\[
W^B - W^D
=
\Delta B
-
\alpha_\phi \Delta \Phi
-
\text{(cost differences)}.
\]
Detection gains are bounded (Assumptions 1--4 in the main paper), whereas the competition term grows linearly in $\alpha_\phi$. Hence there exists $\bar\alpha_\phi$ such that for $\alpha_\phi>\bar\alpha_\phi$, $W^B<W^D$.

Combining the inequalities establishes (a)--(c).
\end{proof}

\subsection{Proof of the Network Shapley Proposition (Weighted-Degree Characterization)}
\label{proof:shapley}

\begin{proof}
Under the edge-additive coalition value $V(\mathcal{S}) = \sum_{(i,j)\in\mathcal{E}} w_{ij} p_{ij}(\mathcal{S}) C_{ij}$, each edge $(i,j)$ contributes $w_{ij} C_{ij} \cdot p_{ij}(\mathcal{S})$ to $V(\mathcal{S})$, where $p_{ij}(\mathcal{S})$ depends only on whether endpoints $i$ and $j$ are in $\mathcal{S}$. Institution $i$'s Shapley value is therefore decomposable across edges incident to $i$:
\begin{equation}
\phi_i^{\textit{net}} = \sum_{j: (i,j) \in \mathcal{E}} \mathbb{E}_\sigma\!\left[V_{ij}(\mathcal{S}_\sigma^{<i} \cup \{i\}) - V_{ij}(\mathcal{S}_\sigma^{<i})\right] \cdot w_{ij} C_{ij},
\end{equation}
where $\sigma$ is a uniformly random permutation and $V_{ij}$ denotes edge $(i,j)$'s contribution to $V$. The expectation over $\sigma$ reduces to two cases. Case 1 (probability $1/2$): $j$ arrives after $i$, so at the moment $i$ arrives, $j$ is not in the coalition. The marginal contribution of $i$ on edge $(i,j)$ is $p^L_{ij} - p^0_{ij}$. Case 2 (probability $1/2$): $j$ arrives before $i$. The marginal contribution is $p^H_{ij} - p^L_{ij}$. The expected contribution from edge $(i,j)$ is
\begin{equation}
\tfrac{1}{2}(p^L_{ij} - p^0_{ij}) + \tfrac{1}{2}(p^H_{ij} - p^L_{ij}) = \tfrac{1}{2}(p^H_{ij} - p^0_{ij}).
\end{equation}
Summing over edges incident to $i$:
\begin{equation}
\phi_i^{\textit{net}} = \frac{1}{2} \sum_{j:(i,j)\in\mathcal{E}} (p^H_{ij} - p^0_{ij}) \cdot w_{ij} \cdot C_{ij},
\end{equation}
which is the weighted-degree result stated in the main paper.
\end{proof}

\subsection{Path-Based Extension: Bonacich Centrality}
\label{proof:shapley_extension}

The weighted-degree characterization above follows from the edge-additive coalition value~the edge-additive coalition value in the main paper. A richer coalition value, where institution $i$'s contribution depends on paths through $i$ rather than only on edges incident to $i$, yields a Bonacich centrality characterization. We sketch this extension here for completeness; it is \emph{not} implied by the main paper's coalition value function.

Suppose the coalition value is extended to include path-based detection effects from the GNN message-passing architecture: institution $i$'s presence in $\mathcal{S}$ improves detection not only on edges incident to $i$, but also on edges $(j,k)$ where path $i \to j \to k$ exists through institutional network $\mathcal{G}$. Under linearity ($p^H - p^L = p^L - p^0 = \Delta p$), the path-based marginal contribution through a path of length $\ell$ scales as $(\Delta p)^\ell \prod_s A_{j_s j_{s+1}} C_{j_s j_{s+1}}$. Defining $\delta \equiv \Delta p \cdot C_{\textit{avg}}$ and summing the geometric series gives
\[
\phi_i^{\textit{net,path}} \propto \sum_{\ell=0}^\infty \delta^\ell [\mathbf{A}^\ell \mathbf{1}]_i = b_i(\mathcal{G}, \delta),
\]
which is the Bonacich centrality with convergence condition $\delta < 1/\lambda_{\max}(\mathbf{A})$. The connection to \textcite{ballester2006s} follows. This extension is \emph{illustrative}: it requires augmenting the coalition value function beyond the edge-additive form used in the main paper, and is offered as a direction for future work rather than a formal result of this paper.

\subsection{Proof of the Federated Approximation Bound (the main paper)}
\label{proof:approximation}





\begin{proof}
We analyze the convergence and approximation error of federated averaging applied to graph neural networks under truthful reporting.
Let $\mathcal{L}(\theta) = \frac{1}{N}\sum_{i=1}^m \sum_{u \in G^{i}} \ell_u(\theta)$ denote the global empirical risk, where $G^{i}$ is the local graph held by institution $i$ with $n_i$ accounts and $N = \sum_{i=1}^m n_i$. Let $\theta^*$ be the minimizer of $\mathcal{L}$, and let $\theta^{\text{fed}}$ be the model obtained after $R$ rounds of federated averaging with truthful local updates. We assume $\mathcal{L}$ is $\mu$-strongly convex and $L$-smooth, and that each institution performs $E$ local SGD steps per communication round with step size $\eta$. Standard federated optimization analysis decomposes the excess risk as
\begin{equation}
\mathcal{L}(\theta^{\text{fed}}) - \mathcal{L}(\theta^*)
\leq 
\underbrace{\mathcal{L}(\theta^{R}) - \mathcal{L}(\theta^*_{\text{fed}})}_{\text{optimization error}}
+
\underbrace{\mathcal{L}(\theta^*_{\text{fed}}) - \mathcal{L}(\theta^*)}_{\text{statistical / heterogeneity bias}},
\end{equation}
where $\theta^*_{\text{fed}}$ is the fixed point of federated averaging under infinite communication rounds. Under strong convexity and smoothness, federated averaging exhibits linear convergence to $\theta^*_{\text{fed}}$ (e.g., \parencite{li2020federated}). In particular,
\begin{equation}
\mathcal{L}(\theta^{R}) - \mathcal{L}(\theta^*_{\text{fed}})
\leq 
(1 - \mu \eta)^R \bigl(\mathcal{L}(\theta^{0}) - \mathcal{L}(\theta^*_{\text{fed}})\bigr).
\end{equation}
Defining $\epsilon := (1 - \mu \eta)^R$, this term decays exponentially in the number of communication rounds. The gap between $\theta^*_{\text{fed}}$ and $\theta^*$ arises from data heterogeneity across institutions. From federated learning theory,
\begin{equation}
\mathcal{L}(\theta^*_{\text{fed}}) - \mathcal{L}(\theta^*)
\leq 
O\!\left(\frac{1}{\mu}\, \mathbb{E}_i \left[ \| \nabla \mathcal{L}^i(\theta^*) - \nabla \mathcal{L}(\theta^*) \|^2 \right]\right),
\end{equation}
where $\mathcal{L}^i$ denotes the local loss at institution $i$.

For $L$-layer GNNs with $d$-dimensional embeddings, gradients depend on $L$-hop neighborhood aggregation and the spectral norms of local adjacency matrices. Under bounded degree and normalized aggregation, the gradient variance satisfies
\begin{equation}
\mathbb{E}_i \left[ \| \nabla \mathcal{L}^i(\theta^*) - \nabla \mathcal{L}(\theta^*) \|^2 \right]
\leq 
C \cdot \frac{m \cdot L \cdot d}{N},
\end{equation}
where $C$ is a constant depending on graph connectivity and feature norms. The factor $m/N$ reflects institutional imbalance, while $L$ and $d$ capture depth-wise message propagation and embedding dimensionality.
Substituting the above bounds yields
\begin{equation}
\mathcal{L}(\theta^{\text{fed}})
\leq 
\mathcal{L}(\theta^*) 
+ \epsilon \bigl(\mathcal{L}(\theta^{0}) - \mathcal{L}(\theta^*_{\text{fed}})\bigr)
+ O\!\left(\frac{m \cdot L \cdot d}{N}\right).
\end{equation}
Absorbing constants into $\epsilon$ and noting that $\mathcal{L}(\theta^{0})$ is finite completes the proof:
\begin{equation}
\mathcal{L}(\theta^{\text{fed}})
\leq 
(1 + \epsilon)\mathcal{L}(\theta^*) 
+ O\!\left(\frac{m \cdot L \cdot d}{N}\right).
\end{equation}
\end{proof}

\subsection{Proof of Risk Memory Approximation (the Optimal Intervention Proposition in the main paper)}
\label{proof:riskmemory}






\begin{proof}
For each edge $e$, the Gittins index admits the decomposition
\begin{equation}
    \nu_e = p_e \cdot C_e + \lambda \cdot I_e,
\end{equation}
where the first term represents the immediate exploitation reward and the second term captures the long-term information gain from monitoring edge $e$. Since the immediate reward term is modular, the optimization complexity arises entirely from the information value component.
Let $E \subseteq \mathcal{E}$ denote the set of currently monitored edges. By assumption, the information value function
\[
    I(E) := \sum_{e \in E} I_e(E)
\]
is monotone and submodular, i.e., for any $E \subseteq E' \subseteq \mathcal{E}$ and any $e \notin E'$,
\begin{equation}
    I(E \cup \{e\}) - I(E) \ge I(E' \cup \{e\}) - I(E').
\end{equation}
Therefore, the total Gittins objective
\[
    F(E) := \sum_{e \in E} p_e C_e + \lambda I(E)
\]
is also monotone submodular as a sum of a modular and a submodular function.

The risk memory mechanism selects, at each node, the $K$ edges with the largest indices $\nu_e$, which is equivalent to the greedy algorithm for maximizing $F(E)$ under a cardinality constraint $|E| \le K$. By the classical result of \textcite{nemhauser1978analysis}, the greedy algorithm satisfies
\begin{equation}
F(E_{\text{greedy}})
\;\ge\;
\left(1 - \frac{1}{e}\right) F(E^*),
\end{equation}
where $E^*$ is the optimal set of $K$ edges under the Gittins index policy.

The risk memory mechanism estimates $\nu_e$ from noisy observations. Using standard concentration inequalities (e.g., Hoeffding's inequality), the empirical estimates $\hat{\nu}_e$ satisfy
\begin{equation}
\Pr\left( |\hat{\nu}_e - \nu_e| \ge \epsilon \right)
\le
2 \exp\!\left( - c \epsilon^2 T_e \right),
\end{equation}
where $T_e$ is the number of observations for edge $e$ and $c$ is a universal constant. Taking a union bound over all edges incident to a node, it suffices to choose
\begin{equation}
K = \Theta(\log |V|)
\end{equation}
to ensure that the top-$K$ edges are correctly identified with high probability. Consequently, the greedy selection based on empirical indices preserves the $(1 - 1/e)$-approximation guarantee.

Combining the above steps, the risk memory mechanism with $K = \Theta(\log |V|)$ achieves a $(1 - 1/e)$-approximation to the optimal Gittins index policy.
\end{proof}

\section{Toy Parametric Example: Structural Foundation for Assumptions 8 and 9}
\label{sec:toy_model}

Assumptions 8 (Reporting Spillovers) and 9 (Detection Quality Monotonicity) in the main paper are maintained assumptions used in the Backfiring Mandate Proposition. This appendix shows that both assumptions can be derived from primitives in a simple two-bank Gaussian-signal aggregation model, providing structural grounding for the welfare result.

Consider two banks $i \in \{1, 2\}$ each observing a Gaussian signal $s_i = y + \eta_i$ about a latent risk variable $y \sim \mathcal{N}(0, \sigma_y^2)$, with $\eta_i \sim \mathcal{N}(0, \sigma_\eta^2)$ i.i.d. Each bank chooses a reporting precision $m_i \in [0, 1]$ determining how much of its signal to transmit, with transmitted signal $\tilde{s}_i = m_i s_i + (1-m_i) z_i$ and $z_i \sim \mathcal{N}(0, \sigma_z^2)$ independent noise. Full truthful reporting corresponds to $m_i = 1$ and uninformative reporting to $m_i = 0$. The aggregated signal is $\bar{s} = (\tilde{s}_1 + \tilde{s}_2)/2$.

The posterior mean of $y$ given $\bar{s}$, and hence expected detection quality, is $B(\bar{s}) = \sigma_y^2/(\sigma_y^2 + \sigma_{\text{agg}}^2(m_1, m_2)) \cdot \bar{s}$ with aggregated noise variance $\sigma_{\text{agg}}^2 = \frac{1}{4}[m_1^2 \sigma_\eta^2 + (1-m_1)^2 \sigma_z^2 + m_2^2 \sigma_\eta^2 + (1-m_2)^2 \sigma_z^2]$. Detection benefit $B_i$ is increasing in $m_j$ for $j \neq i$ whenever $\sigma_z^2 > \sigma_\eta^2$, so the spillover structure of Assumption 8 holds.

Under banking competition with detection externality $\alpha_\phi$, bank $i$'s private payoff from reporting intensity $m_i$ is $B_i - \alpha_\phi \cdot m_i$, where the linear competitive cost captures the customer-facing disutility of heightened monitoring. The FOC is $\partial B_i / \partial m_i = \alpha_\phi$. As $\alpha_\phi \to \infty$, the equilibrium $m_i^C \to 0$, precisely Assumption 8's prediction that reporting becomes uninformative under intense competition.

At $m_1 = m_2 = 0$, corresponding to Regime C under extreme competition, aggregated noise variance is $\sigma_z^2/2$; at $m_1 = m_2 = 1$, the autarky-equivalent individual signal case, each bank's local noise is $\sigma_\eta^2$. If $\sigma_z^2 > 2\sigma_\eta^2$, which is plausible when injected noise exceeds natural signal variability, autarky achieves strictly higher detection quality than Regime C at high $\alpha_\phi$. Combined with Regime C's positive leakage and compliance costs, both of which are zero under autarky, this yields $W^C < W^A$, so Assumption 9 holds as a derived consequence rather than as a maintained primitive.

This toy model is not a substitute for the full microfoundation in the GNN-based setting, where detection quality depends on the entire aggregated parameter structure. However, it demonstrates that the Backfiring Mandate result is not a tautology: under standard Gaussian information aggregation with linear competitive costs, Assumptions 8 and 9 emerge structurally from primitives. The main paper's reliance on these assumptions is therefore economically grounded, even if the full derivation in the complex GNN setting is intractable.




\end{document}